\documentclass[final]{IEEEtran}
\usepackage{amsthm,amssymb,graphicx,multirow,amsmath,color,amsfonts,physics}%,ulem}
\usepackage[update,prepend]{epstopdf}
\usepackage[noadjust]{cite}
\usepackage[latin1]{inputenc}
\usepackage{tikz}
\usepackage{bbm} % for \nbb1 
\usepackage{pdfpages}
\usepackage{balance}
\usepackage{multirow}
\usepackage{comment}
\usepackage{subfigure}
\allowdisplaybreaks % Allows breaking of eqnarray over multiple pages (avoids unnecessary blanks in the document before eqnarray)

\begin{document}
% Bold lowercase: syntax \nb# where # is {a ... z, 0,1}
\def\nba{{\mathbf{a}}}
\def\nbb{{\mathbf{b}}}
\def\nbc{{\mathbf{c}}}
\def\nbd{{\mathbf{d}}}
\def\nbe{{\mathbf{e}}}
\def\nbf{{\mathbf{f}}}
\def\nbg{{\mathbf{g}}}
\def\nbh{{\mathbf{h}}}
\def\nbi{{\mathbf{i}}}
\def\nbj{{\mathbf{j}}}
\def\nbk{{\mathbf{k}}}
\def\nbl{{\mathbf{l}}}
\def\nbm{{\mathbf{m}}}
\def\nbn{{\mathbf{n}}}
\def\nbo{{\mathbf{o}}}
\def\nbp{{\mathbf{p}}}
\def\nbq{{\mathbf{q}}}
\def\nbr{{\mathbf{r}}}
\def\nbs{{\mathbf{s}}}
\def\nbt{{\mathbf{t}}}
\def\nbu{{\mathbf{u}}}
\def\nbv{{\mathbf{v}}}
\def\nbw{{\mathbf{w}}}
\def\nbx{{\mathbf{x}}}
\def\nby{{\mathbf{y}}}
\def\nbz{{\mathbf{z}}}
\def\nb0{{\mathbf{0}}}
\def\nb1{{\mathbf{1}}}

% Bold capital letters: syntax \nb# where # is {A ... Z}
\def\nbA{{\mathbf{A}}}
\def\nbB{{\mathbf{B}}}
\def\nbC{{\mathbf{C}}}
\def\nbD{{\mathbf{D}}}
\def\nbE{{\mathbf{E}}}
\def\nbF{{\mathbf{F}}}
\def\nbG{{\mathbf{G}}}
\def\nbH{{\mathbf{H}}}
\def\nbI{{\mathbf{I}}}
\def\nbJ{{\mathbf{J}}}
\def\nbK{{\mathbf{K}}}
\def\nbL{{\mathbf{L}}}
\def\nbM{{\mathbf{M}}}
\def\nbN{{\mathbf{N}}}
\def\nbO{{\mathbf{O}}}
\def\nbP{{\mathbf{P}}}
\def\nbQ{{\mathbf{Q}}}
\def\nbR{{\mathbf{R}}}
\def\nbS{{\mathbf{S}}}
\def\nbT{{\mathbf{T}}}
\def\nbU{{\mathbf{U}}}
\def\nbV{{\mathbf{V}}}
\def\nbW{{\mathbf{W}}}
\def\nbX{{\mathbf{X}}}
\def\nbY{{\mathbf{Y}}}
\def\nbZ{{\mathbf{Z}}}

% \mathcal: syntax \ncal# where # is {A ... Z}
\def\ncalA{{\mathcal{A}}}
\def\ncalB{{\mathcal{B}}}
\def\ncalC{{\mathcal{C}}}
\def\ncalD{{\mathcal{D}}}
\def\ncalE{{\mathcal{E}}}
\def\ncalF{{\mathcal{F}}}
\def\ncalG{{\mathcal{G}}}
\def\ncalH{{\mathcal{H}}}
\def\ncalI{{\mathcal{I}}}
\def\ncalJ{{\mathcal{J}}}
\def\ncalK{{\mathcal{K}}}
\def\ncalL{{\mathcal{L}}}
\def\ncalM{{\mathcal{M}}}
\def\ncalN{{\mathcal{N}}}
\def\ncalO{{\mathcal{O}}}
\def\ncalP{{\mathcal{P}}}
\def\ncalQ{{\mathcal{Q}}}
\def\ncalR{{\mathcal{R}}}
\def\ncalS{{\mathcal{S}}}
\def\ncalT{{\mathcal{T}}}
\def\ncalU{{\mathcal{U}}}
\def\ncalV{{\mathcal{V}}}
\def\ncalW{{\mathcal{W}}}
\def\ncalX{{\mathcal{X}}}
\def\ncalY{{\mathcal{Y}}}
\def\ncalZ{{\mathcal{Z}}}

% \mathbb: syntax \nbb# where # is {A ... Z}
\def\nbbA{{\mathbb{A}}}
\def\nbbB{{\mathbb{B}}}
\def\nbbC{{\mathbb{C}}}
\def\nbbD{{\mathbb{D}}}
\def\nbbE{{\mathbb{E}}}
\def\nbbF{{\mathbb{F}}}
\def\nbbG{{\mathbb{G}}}
\def\nbbH{{\mathbb{H}}}
\def\nbbI{{\mathbb{I}}}
\def\nbbJ{{\mathbb{J}}}
\def\nbbK{{\mathbb{K}}}
\def\nbbL{{\mathbb{L}}}
\def\nbbM{{\mathbb{M}}}
\def\nbbN{{\mathbb{N}}}
\def\nbbO{{\mathbb{O}}}
\def\nbbP{{\mathbb{P}}}
\def\nbbQ{{\mathbb{Q}}}
\def\nbbR{{\mathbb{R}}}
\def\nbbS{{\mathbb{S}}}
\def\nbbT{{\mathbb{T}}}
\def\nbbU{{\mathbb{U}}}
\def\nbbV{{\mathbb{V}}}
\def\nbbW{{\mathbb{W}}}
\def\nbbX{{\mathbb{X}}}
\def\nbbY{{\mathbb{Y}}}
\def\nbbZ{{\mathbb{Z}}}

% \mathfrak:
\def\nfrakR{{\mathfrak{R}}}

% Roman: {\rm } syntax \nrm# where # is {a ... z}
\def\nrma{{\rm a}}
\def\nrmb{{\rm b}}
\def\nrmc{{\rm c}}
\def\nrmd{{\rm d}}
\def\nrme{{\rm e}}
\def\nrmf{{\rm f}}
\def\nrmg{{\rm g}}
\def\nrmh{{\rm h}}
\def\nrmi{{\rm i}}
\def\nrmj{{\rm j}}
\def\nrmk{{\rm k}}
\def\nrml{{\rm l}}
\def\nrmm{{\rm m}}
\def\nrmn{{\rm n}}
\def\nrmo{{\rm o}}
\def\nrmp{{\rm p}}
\def\nrmq{{\rm q}}
\def\nrmr{{\rm r}}
\def\nrms{{\rm s}}
\def\nrmt{{\rm t}}
\def\nrmu{{\rm u}}
\def\nrmv{{\rm v}}
\def\nrmw{{\rm w}}
\def\nrmx{{\rm x}}
\def\nrmy{{\rm y}}
\def\nrmz{{\rm z}}

% Special symbols
\def\nbydef{:=}
\def\nborel{\ncalB(\nbbR)}
\def\nboreld{\ncalB(\nbbR^d)}
\def\sinc{{\rm sinc}}

% Theorems etc.
\newtheorem{lemma}{Lemma}
\newtheorem{thm}{Theorem}
\newtheorem{definition}{Definition}
\newtheorem{ndef}{Definition}
\newtheorem{nrem}{Remark}
\newtheorem{theorem}{Theorem}
\newtheorem{prop}{Proposition}
\newtheorem{cor}{Corollary}
\newtheorem{example}{Example}
\newtheorem{remark}{Remark}
\newtheorem{assumption}{Assumption}
	
%%%%%%%% Backwards compatibility

\newcommand{\ceil}[1]{\lceil #1\rceil}
\def\argmin{\operatorname{arg~min}}
\def\argmax{\operatorname{arg~max}}
\def\figref#1{Fig.\,\ref{#1}}%
\def\E{\mathbb{E}}
\def\EE{\mathbb{E}^{!o}}
\def\P{\mathbb{P}}
\def\pc{\mathtt{P_c}}
\def\rc{\mathtt{R_c}}   % rate coverage
\def\p{p}

\def\V{\operatorname{Var}}
\def\erfc{\operatorname{erfc}}
\def\erf{\operatorname{erf}}
\def\opt{\mathrm{opt}}
\def\R{\mathbb{R}}
\def\Z{\mathbb{Z}}

\def\LL{\mathcal{L}^{!o}}
\def\var{\operatorname{var}}
\def\supp{\operatorname{supp}}

\def\N{\sigma^2}
\def\T{\beta}							% Threshold = \beta_i
\def\sinr{\mathtt{SINR}}			% Signal to interference plus noise ratio
\def\snr{\mathtt{SNR}}
\def\sir{\mathtt{SIR}}
\def\ase{\mathtt{ASE}}
\def\se{\mathtt{SE}}

\def\calN{\mathcal{N}}
\def\FE{\mathcal{F}}
\def\calA{\mathcal{A}}
\def\calK{\mathcal{K}}
\def\calT{\mathcal{T}}
\def\calB{\mathcal{B}}
\def\calE{\mathcal{E}}
\def\calP{\mathcal{P}}
\def\calL{\mathcal{L}}
%\DeclareMathOperator{\Tr}{Tr}
%\DeclareMathOperator{\rank}{rank}
%\DeclareMathOperator{\Pois}{Pois}

%\DeclareMathOperator{\TC}{\mathtt{TC}}
%\DeclareMathOperator{\TCL}{\mathtt{TC_l}}
%\DeclareMathOperator{\TCU}{\mathtt{TC_u}}

% Fading
\def\l{\ell}
\newcommand{\fad}[2]{\ensuremath{\mathtt{h}_{#1}[#2]}}
\newcommand{\h}[1]{\ensuremath{\mathtt{h}_{#1}}}

\newcommand{\err}[1]{\ensuremath{\operatorname{Err}(\eta,#1)}}
\newcommand{\FD}[1]{\ensuremath{|\mathcal{F}_{#1}|}}

%% Symbols changed
% \def\i{\mathbf{1}}					% changed to \nb1
% \def\d{\mathrm{d}}					% changed to \nrmd
% \def\L{\mathcal{L}}					% changed to \ncalL
% \begin{definition}					% changed to \begin{ndef}

% \l also gives problems. Use \ell after defining it if needed.

%% D2D def
\def\Bx{{\mathcal{B}}^x}
\def\Bxx{{\mathcal{B}}^{x_0}}
\def\jx{y}
\def\m{(\bar{n}-1)}
\def\mm{\bar{n}-1}
\def\Nx{{\mathcal{N}}^x}
\def\Nxo{{\mathcal{N}}^{x_0}}
\def\wj{w_{jx_0}}
\def\uij{u_{jx}}
 \def\yj{y}
 \def\yjx{y}
 \def\zjx{z_x}
 \def \tx {y_0}
 \def \htx {h_0}

\def\rx{z_{1}}
\def\ry{z_{2}}

\def\Rx{Z_{1}}
\def\Ry{Z_{2}}

%% fading
\def \hyxx {h_{y_{x_0}}}
\def \hyx {h_{y_x}}

\def\nbb1{\mathbbm{1}}
\def\yi{\textbf{y}_i}
\def\yy{\textbf{y}_1}
\def\xx{\textbf{x}_0}
\def\wj{\textbf{w}_j}
\def\od{\textbf{o}_d}
\def\oc{\textbf{o}_c}
\def\ie{{\em i.e. }}
\def\eg{{\em e.g. }}
\def\iid{{\em i.i.d. }}
\def\gi{G_{\yi}}
\def\g1{G_{\yy}}
\def\gj{G_{\wj}}
\def\hi{H_{\yi}}
\def\h1{H_{\yy}}
\def\hj{H_{\wj}}
\def\avg{\rm avg}

\def\rmnuma{\rm\uppercase\expandafter{\romannumeral1}}
\def\rmnumb{\rm\uppercase\expandafter{\romannumeral2}}
\def\rmnumc{\rm\uppercase\expandafter{\romannumeral3}}
\def\rmnumd{\rm\uppercase\expandafter{\romannumeral4}}
\def\rmnume{\rm\uppercase\expandafter{\romannumeral5}}
\def\rmnumf{\rm\uppercase\expandafter{\romannumeral6}}
\pagenumbering{gobble}
\graphicspath{{./Figures/}}
\title{
Performance Evaluation of RF-powered\\ IoT in Rural Areas:\\ The Wireless Power Digital Divide}
\author{
 Hao Lin,~\IEEEmembership{Student Member,~IEEE}, Mustafa A. Kishk,~\IEEEmembership{Member,~IEEE} and Mohamed-Slim Alouini,~\IEEEmembership{Fellow,~IEEE}
\thanks{Hao Lin is with the Electrical and Computer Engineering Program, Computer, Electrical and Mathematical Sciences and Engineering Division (CEMSE), King Abdullah University of Science and Technology (KAUST),
Thuwal 23955-6900, Saudi Arabia (e-mail: hao.lin.std@gmail.com).\\
% Hao Lin is with the University of Electronic Science and Technology
% of China (UESTC), Chengdu 611731, China, and also with the CEMSE
% Division, King Abdullah University of Science and Technology (KAUST),
% Thuwal 23955-6900, Saudi Arabia (e-mail: hao.lin.std@gmail.com).\\
\indent Mustafa A. Kishk is with the Department of Electronic Engineering Program,
Maynooth University, Maynooth, W23 F2H6 Ireland (e-mail:
mustafa.kishk@mu.ie).\\%National University of Ireland
\indent Mohamed-Slim Alouini is with the CEMSE Division, King Abdullah
University of Science and Technology (KAUST), Thuwal 23955-6900,
Saudi Arabia (e-mail: slim.alouini@kaust.edu.sa).}
}

\maketitle
\vspace{-2cm}
\begin{abstract}
Bridging the digital divide is one of the goals of mobile networks in the future, and further building IoT networks in rural areas is a feasible solution. This paper studies the downlink performance of rural wireless networks, where IoT devices we consider are battery-less and powered only by ambient radio-frequency (RF) signals. We model a rural area as a finite area that is far from the city center. The base stations (BSs) in the whole city and the access points (APs) in the finite network both act as sources of wireless RF signals harvested by IoT devices. We assume that BSs follow an inhomogeneous Poisson Point Process (PPP) with a 2D-Gaussian density, and a fixed number of APs are uniformly distributed inside the finite area following a Binomial Point Process (BPP). The IoT devices we consider can harvest energy and receive downlink signals in each time slot, which is divided into two parts: (1) a charging sub-slot, where the RF signals from BSs and APs are harvested by IoT devices, and (2) a transmission sub-slot, where each IoT device uses the harvested energy to receive and process downlink signals. We consider two main system requirements: minimum energy requirement and signal-to-interference-plus-noise ratio (SINR). Using these two parameters, we investigate the overall coverage probability (OCP) related to them. We first study the effect of remoteness in rural areas on energy harvesting performance. Then we analyze the influence of IoT device's location and the number of APs on coverage probability when the effect of BSs can be ignored. This paper shows that the IoT devices located inside the rural area can obtain about twice the ECP and OCP of IoT devices located near the edge. For the average downlink performance in rural areas with radii less than $100\;{\rm m}$, more than 80\% of the RF-powered IoT devices can be supported when there are 100 APs deployed.
\end{abstract}
% \vspace{-0.7cm}
\begin{IEEEkeywords}
% \vspace{-0.3cm}
Stochastic geometry, finite wireless network, energy harvesting, Binomial Point Process, overall coverage probability.
\end{IEEEkeywords}

\section{Introduction} \label{sec:Intro}
\indent Although 5G cellular networks are gradually being implemented and 6G key technologies are also being studied, most people in rural areas still face challenges when using the Internet. Due to the limitation of population density and economic benefits, the communication infrastructure deployed in rural areas around the world has lagged behind that in urban areas \cite{9681630}, which has caused the digital divide. The digital divide is mainly seen as a challenge, but also as an opportunity for further development of wireless networks. Many methods for understanding and eliminating the digital divide have been proposed, including data-driven policy-making \cite{9681629}, high altitude platform stations \cite{9900369}, mega configurations of low earth orbit (LEO) \cite{9681631}, wind-turbine-mounted base stations (WTBS) \cite{9681626}, and cellular-connected unmanned aerial vehicles (UAVs) toward 6G \cite{mozaffari2021toward}. It is worth noting that, it is a feasible choice to use Internet of Things (IoT) devices to help build modern agricultural industry technology systems and bridge the digital divide in the whole world.\\
\indent The core concept of the Internet of Things is to equip everyday objects with the functions of identification, sensing, networking, and processing so that people can communicate or interact with each other through the Internet to achieve the goals they expect \cite{whitmore2015internet}. The application scenarios of IoT devices include modern medical care, urban transportation, smart home, smart agriculture, and other scenarios with development needs and potential \cite{patel2016internet}. The effective use of IoT also contributes to the achievement of the United Nations Sustainable Development Goals (SDGs). Unlike the devices that people carry with them every day like mobile phones and tablets, IoT devices may be deployed in some areas hard to be directly reached, such as inside instruments or equipment and remote locations like forests or farms \cite{7842431,6644241}. Their spatial location makes wired charging and battery replacement difficult while harvesting radio frequency (RF) signals becomes a viable option \cite{6951347,7120022}. Authors in \cite{7010878} summarized contributions of energy harvesting wireless communications (EHWC), including offline/online energy management, joint wireless energy and information transfer, and the energy consumption models. Communication networks with simultaneous wireless information and power transfer (SWIPT) is an important concept where the signal sources also work as energy stations of RF-powered devices. A general receiver operation has been proposed in literature \cite{6623062}, which can split the received signals for energy charging and information decoding, to study the trade-off between the energy demand and the data rate. \\
\indent Stochastic geometry framework is widely used in the performance evaluation of wireless networks. Poisson Point Process (PPP) and Binomial Point Process (BPP) are both used to model the distribution of RF sources and IoT devices, which are suitable in infinite space and finite space, respectively \cite{7882710}.  %Authors in \cite{okati2020downlink} and \cite{cherif2020downlink} both model the distribution of satellite networks and aerial-BSs as BPPs. The work in \cite{kishk2020exploiting} model the locations of BSs being visible by a typical user by an inhomogeneous PPP in a large-scale wireless network with reconfigurable intelligent surfaces (RISs). 
In this paper, we investigate the performance of RF-powered IoT devices in rural and remote areas while capturing the influence of the lack of infrastructure in such regions. More details on the contributions are provided in Sec. \ref{subsec:contributions}.

\subsection{Related Work}
Stochastic geometry is an important mathematical tool, which can help us tractably analyze various types of wireless networks without losing much accuracy. Energy harvesting (EH), especially from ambient wireless signals, can become particularly necessary in rural and remote areas when it is impractical to rely on frequent battery replacement. In this article, we mainly study the effect of relying on energy harvesting for IoT on the communication performance in rural areas. Therefore, we categorize relevant literature into {\em i) stochastic geometry for energy harvesting} and {\em ii) communications in rural areas}.\\
\indent {\em Stochastic geometry for energy harvesting}: Energy harvesting wireless networks have many benefits, including increasing the lifetime of devices and reducing reliance on traditional energy sources \cite{6786061}. For different energy harvesting models, authors in \cite{9201540} characterized the harvested energy and average achievable downlink rate for SWIPT in cell-free massive multiple-input multiple-output (MIMO). Authors in \cite{9723559} investigated the application of SWIPT in networks with cooperative non-orthogonal multiple access (CNOMA) and derived the expression of outage probability and throughput based on PPP assumption.   
 In literature \cite{8536464}, joint energy and signal-to-interference-plus-noise ratio (SINR) coverage probability was proven useful for evaluating the performance of RF-powered devices with a time-slotted architecture. In addition, authors in \cite{8896898} studied the end-to-end outage probability in EH-enabled device-to-device (D2D) cellular networks, and the network performance for cache- and EH-enabled D2D enabled cellular networks was investigated in literature \cite{9390298}. To maximize coverage and EH probabilities, authors in \cite{8888216} optimized the altitude of RF-powered aerial base stations serving terrestrial devices based on a stochastic geometry analysis. In addition, the work in \cite{Fatma_EH, 9093022, 9463400} characterized the outage probability and self-sustainability of RF-powered D2D IoT networks. 
\indent {\em Communications in rural areas}: Several concepts and theories about communications in rural areas have been introduced in \cite{demiryurek2010information}, including information systems and communication networks for agriculture and people in rural areas. The economic, social, educational, and knowledge inequalities between those with information and communication technology (ICT) and those without ICT were defined as \textit{digital divide}, and many rural communities in the world are seeking solutions to bridge the digital divide \cite{el2020toward,chaoub20216g}. After the 5G standard was proposed, the challenges and opportunities of future rural wireless communications have been discussed, where novel low-cost network enablers (including non-terrestrial networks and drone base stations) were proved feasible to support future applications in rural areas such as smart agriculture \cite{9681630}. Authors in \cite{khalil2017feasibility} analyzed the feasibility, architecture, and cost considerations of using TV white spaces (TVWS) to enhance rural Internet access in 5G cellular networks. The work on \cite{katzis2020opportunities} presented the advantages and challenges of employing TVWS spectrum with 5G enabled high altitude platforms (HAPs). In \cite{qin2022drone}, renewable energy charging stations were demonstrated which address the limited battery of EH-UAVs. A new stochastic geometry framework was presented in \cite{9420290} to evaluate the coverage probability in a region including urban and rural areas, and it introduced aerial base stations (ABSs) to compensate for the deficiency in such a UAV-assisted cellular network. Other studies such as \cite{9681626} proved that it is feasible to use wind turbine-mounted base stations (WTBSs) to enhance rural connectivity. Authors in \cite{fourati2022bridging} suggested that artificial intelligence (AI) technologies, especially reinforcement learning (RL) and federated learning (FL), are expected to become the key to space, air, and ground optimization and integration for rural connectivity.\\
\indent {However, existing works lack a unified mathematical framework for communication networks in a large-scale city with urban areas and rural areas and the overall performance of RF-powered IoT devices. There are few discussions of the constraints for RF-powered IoT networks in rural areas.}

\subsection{Contributions}\label{subsec:contributions}
The contributions of this paper can be summarized as follows:
\begin{itemize}
    \item {\em System modeling and problem statement}: We design a large-scale city model to describe a novel terminology, which is the \textit{wireless power digital divide}, that refers to the unbalanced capabilities to harvest wireless energy between IoT devices in urban and rural areas. For that purpose, we propose a system model composed of a cellular network with BSs distributed throughout the city according to an inhomogeneous PPP, to capture the reduction in BS density in rural areas. {The proposed system also involves a fixed number of APs deployed in a finite area and modeled as a BPP, while the rural area is modeled as a finite area far away from the city center.} The proposed system enables mathematically capturing the unfairness in access to wireless power between urban and rural users. %To describe the digital divide between rural areas and urban areas, we model the BSs in the whole city as an inhomogeneous PPP with a 2D-Gaussian density like \cite{9420290}. Differently, we model a fixed number of APs in the finite area using a BPP considering the space and available resources are limited in it. We adopt the 'time-slotted' architecture of our IoT device and divide each time slot into two sub-slots: {\em i) a charging sub-slot} and {\em ii) a transmission sub-slot}. In the charging sub-slot, the IoT device harvests all the ambient RF signals from BSs and APs with Rayleigh fading and standard power law path-loss. In the transmission sub-slot, we consider the signal from the closest AP is useful while the signals from other APs and BSs are considered as interference. Based on the system model, we introduce the $k^{th}$ closest TX-selection policy in \cite{7248843} and set $k=1$ to calculate the harvested energy and {\rm SINR}. Both considering energy coverage probability (ECP) and transmission coverage probability (TCP), we define the overall coverage probability (OCP) by jointing them together (\ie the joint energy and SINR coverage probability). Our first goal is to analyze how remoteness in rural areas influences energy harvesting performance. When the rural area is far from the city center, the effect of BSs can be ignored and the APs are the only source of RF signals. In this case, we discuss the effect of IoT device's location and the number of APs. If the IoT devices are uniformly distributed inside the rural area, we further propose a `center-edge distribution' and analyze the existence of optimal area partitioning of APs when IoT device also follows this special distribution.\\
    \item {\em Technical and mathematical challenges}: This paper is the first one to compute the distribution of wireless harvested energy with the locations of RF sources distributed as a superposition of an inhomogeneous PPP and a BPP. Given that RF sources are distributed according to different point processes, a level of complexity arises in both energy coverage analysis as well as interference analysis, which are both handled carefully in this paper. In addition, as a result of the assumption of both RF-powered IoT devices and local RF sources to be distributed in a finite area, few novel distance distributions are derived in this paper and properly utilized for the performance analysis of the considered system. %Because APs inside the rural area and BSs located in the whole city area have different uniform power of ambient RF signals and distributions, we could not analyze them using the same method. The harvested energy from APs can be regarded as the product of uniform power and summation of fadings, and the energy harvested from BSs is similar. We could not calculate them using the distributions directly, so we require an available method and the closest TX-selection policy is suitable. We define the $D_i$ which means the Euler distance between the APs and the IoT device. We use the closest $D_i$ to calculate not only the conditional ECP but also TCP. Adopting the closest TX-selection policy, the closest AP to our receiver is considered as the serving one whose signal is useful, and others are all considered as interference at the same time. To analyze the transmission coverage probability, we should use the Laplace transform of interference from disturbing APs or BSs. All signal sources are independently distributed, but APs follow a BPP while BSs follow an inhomogeneous PPP with a 2D-Gaussian density, which increases the computational complexity. We use the PGFL of PPP and double integral in the whole space to achieve the expression of coverage probabilities. When we introduce the `center-edge' distribution of IoT devices and APs, the areas are divided into several parts, which makes the expressions more complex.\\
    \item {\em Main findings and revealed insights}: In this paper, we first derive the expression of harvested energy, {\rm SINR} and coverage probabilities. Then, we show how the remoteness of finite areas affects energy harvesting performance, which causes the \textit{wireless power digital divide}. We also describe the effect of the number of APs and the distance between the IoT device and the rural center on energy harvesting performance and overall performance. Finally, we prove the existence of APs' optimal distribution when the IoT devices and APs inside rural areas are further limited in a special case. %Using the stochastic geometric tools, we talk about the effect of remoteness in rural areas on energy harvesting. When the rural area is close to the city center, energy coverage can be generally satisfied. However, due to the uneven density of BSs in the whole city, the harvesting performance of the receiver located in rural areas declines as the rural area are far away from the city center. Based on the existence of APs, the receiver in rural areas which lost assistance from BSs can still have a stable performance. In this case, the receiver inside the rural area has better overall performance than beside the edge, and the number of APs influences the energy supply and the interference at the same time. Because of spatial constraints, it is difficult for receivers and APs to be completely distributed in all parts of the rural area. We prove the existence and changing trend of the optimal distribution of APs in this case. 
\end{itemize} 
% \indent We have introduced some related work about how stochastic geometry helps us model wireless networks with energy-harvesting receivers. Then, we build the system model and describe the problems we want to discuss in this paper. We mention some technical and mathematical challenges and how we solve them. Finally, we conclude with the main findings of this paper.\\
% System modeling and problem statement

\section{System Model} \label{sec:SysMod}
When we study the energy harvesting performance and overall performance of each IoT device with energy harvesting, we need to consider these factors: the location or distribution of IoT devices, the remoteness of the finite area, and the distribution of BSs in the whole city and APs in the finite area. {In the large city model we introduce in this section, a rural area is considered a finite area far from the city center. We describe the basic system model in Fig.\ref{fig:system}.}
\begin{figure*}[htbp]
    \centering
    \includegraphics[width=1.5\columnwidth]{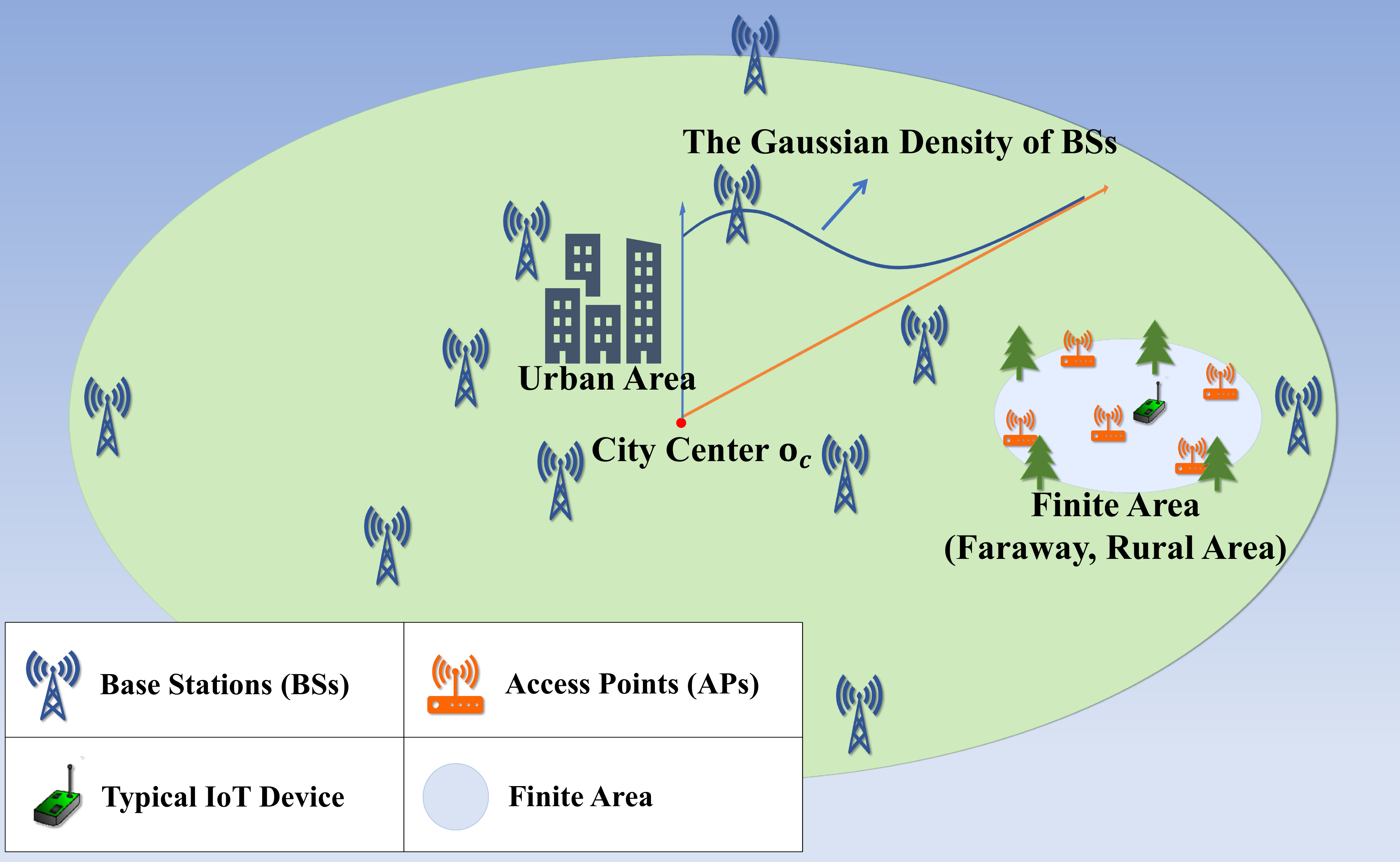}%
    \caption{Illustration of the system model. {There is a finite area in a large city. A typical IoT device (receiver) is located inside this finite area. A finite area that is far from the city center is used to model a rural area.} The BSs in the whole city follow an inhomogeneous PPP with a 2D-Gaussian density $\widetilde{\lambda}_p(\zeta)$ and the APs inside the finite area are distributed according to a BPP.}
    \label{fig:system}
\end{figure*}

\subsection{IoT devices Modeling}
As introduced in \cite{7882710}, we can define a finite area as $\ncalA=\textbf{b}(\od,r_d)\subset \nbbR^2$ which means a circular area centered at $\od$ with radius $r_d$. The BSs and APs both act as sources of wireless RF signals harvested by IoT devices. Next, we model the location of the typical IoT device and distribution of IoT devices (working as receivers) using different coordinate systems in two different cases: \textit{1) considering both BSs and APs} and \textit{2) only considering APs}. The second case is particularly relevant in remote areas with no presence of cellular networks, which is investigated in a more detailed manner. 
\subsubsection{Considering both BSs and APs}
 If the RF signals from BSs can not be ignored, we assume that the rural center $\od$ is on $\mathtt{x-axis}$, \ie $\od=(r_c,0)$ where $r_c$ is the distance between $\od$ and $\oc$. In the finite area, we assume that $\mathtt{x'-axis}$ is the line that passes through the typical IoT device located at $\xx$ and the rural center $\od$. As shown in Fig.\ref{fig:rural}, the angle between $\mathtt{x'-axis}$ and $\mathtt{x-axis}$ is $\psi$, and the coordinates of the typical IoT device location is $\xx=(\zeta_0,\gamma_0)$, where $\zeta_0=\sqrt{r_c^2+v_0^2+2r_cv_0\cos \psi}$ and $\gamma_0=\arctan(\frac{v_0\sin \psi}{r_c+v_0\cos \psi})$. In general, using $\mathtt{x-axis}$ and $\mathtt{x'-axis}$, we can obtain the positional relationship between the typical IoT device, rural center and city center when $r_c$, $v_0$ and $\psi$ are known.\\
 \indent {When the finite area is not far away from the city center, the performance of each IoT device is affected by the cellular network in the whole city (BSs) and the local network in the finite area (APs).} IoT devices in different positions have different energy harvesting performance and overall coverage performance, which are relevant to the value of $r_c$, $v_0$ and {$\psi$}.

\begin{figure}
    \centering
    \includegraphics[width=1\columnwidth]{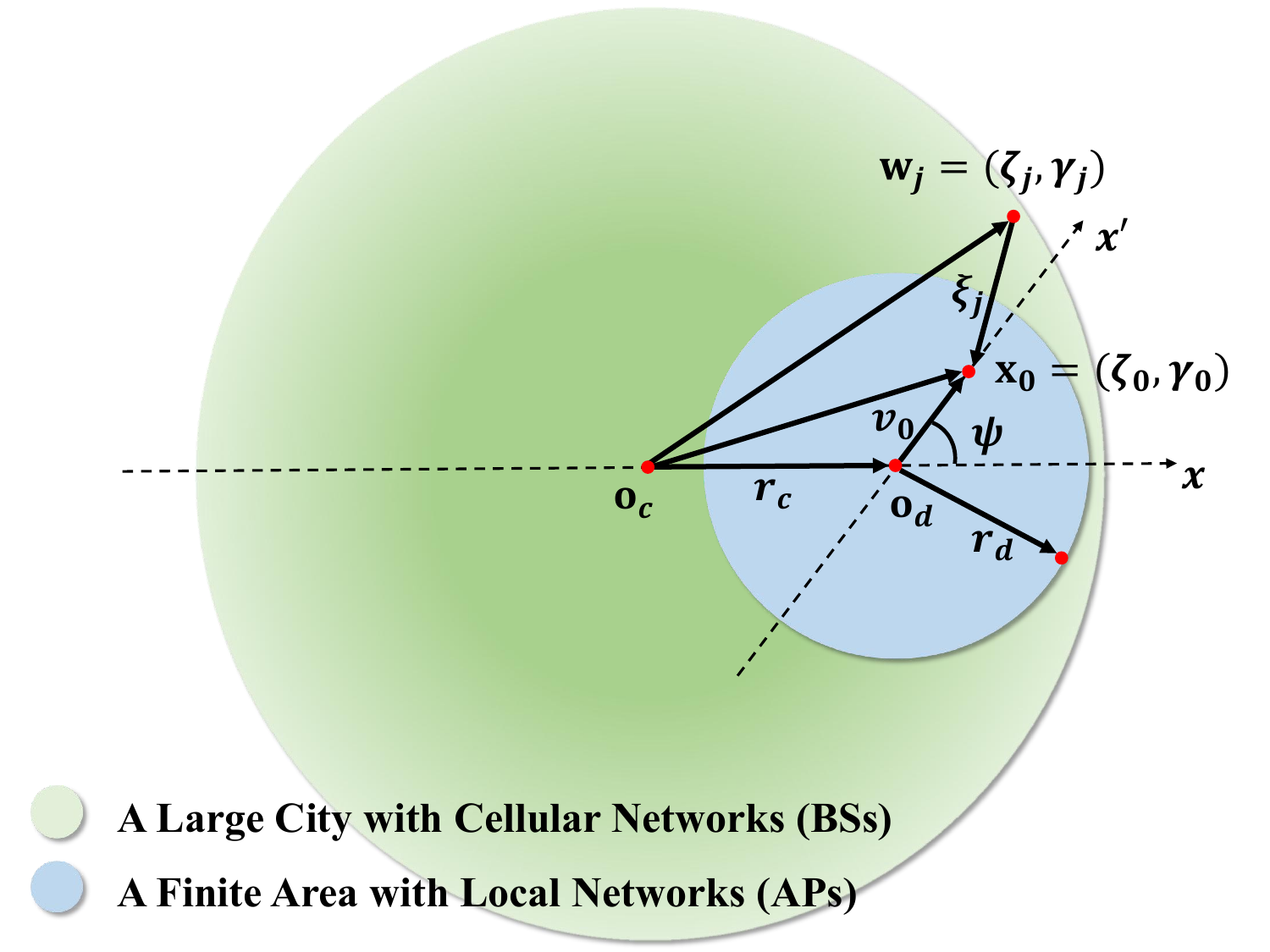}
    \caption{Illustration of the mathematical model of networks in the large city and the considered finite area inside it. {In the polar coordinate system, $\od$ and $r_d$ are the center and radius of the finite area, and $\oc$ is the center of the whole city. The location of any of BSs is $\wj=(\zeta_j,\gamma_j)$ and the location of the typical IoT device (receiver) is $\xx=(\zeta_0,\gamma_0)$.} The distance between them is $\xi_j=\sqrt{\zeta_j^2+\zeta_0^2-2\zeta_j\zeta_0\cos(\gamma_j-\gamma_0)}$.}
    \label{fig:rural}
\end{figure}

\subsubsection{Only considering APs}
{When the finite area is far away from the city center, the effect of BSs on IoT devices can be ignored. Therefore, we only need to consider APs uniformly distributed in the finite area $\ncalA$.} The performance of the typical IoT device only depends on the distance to the center of $\ncalA$ because $\ncalA$ is a circular area and APs are uniformly distributed in it. Such a finite area $\ncalA$ far away from the city center is used to model a rural area, whose center is called `rural center' $\od$. If the typical IoT device location is fixed, we still assume that it is on the $\mathtt{x'-axis}$ and $\xx=(v_0,0)$, where $v_0=\|\xx-\od\| \in\left[0,r_d\right]$ is the distance to the rural center. The CDF of $v_0$ is the ratio of $\pi v_0^2$ and $\pi r_d^2$, and the PDF is the derivative which can be written as:
\begin{equation}
f_V(v_0)=\frac{2v_0}{r_d^2}, 0<v_0<r_d.
\label{fVv}
\end{equation}

\indent In the actual situation, some places in rural areas are not suitable to set up IoT devices or APs, such as rivers or muddy areas. The center and border areas of rural areas are generally more likely to be used for networking, such as residential areas or transportation tracks. So we propose a new distribution called `center-edge distribution', in which the IoT devices are uniformly and randomly distributed in these two areas: \textit{i)} $\ncalA_1=\textbf{b}(\od,R_1)$\textit{, a circular area centered at $\od$ with radius $R_1$,} and \textit{ii)} $\ncalA\backslash\ncalA_2=\textbf{b}(\od,r_d)\backslash\textbf{b}(\od,R_2)$\textit{, a ring area centered at the $\od$ with inner radius $R_2$ and outer radius $r_d$, where $R_2>R_1$}. By setting the values of $R_1$ and $R_2$, we can model finite areas based on different natural environments and planning. The new PDF of the device-to-center distance $v_0$ in this case is expressed as:
\begin{equation}
    % f_V^{CE}(v_0)=\left\{
    % \begin{array}{cl}
    %      \frac{2v}{r_d^2-R_2^2+R_1^2},&0<v_0<R_1 \;or\; R_2<v_0<r_d\\
    %      0,&otherwise
    % \end{array}
    % \right.,
     f_V^{CE}(v_0)=\frac{2v_0}{r_d^2-R_2^2+R_1^2},\;0<v_0<R_1 \;{\rm or}\; R_2<v_0<r_d,
\end{equation}
and we use an ordered pair $(R_1, R_2)$ to represent the range of the IoT device's distribution in this case.

\subsection{Base Stations Modeling}
In a large city, we model the rural area as a limited space where the IoT devices can not only receive RF signals from APs located in this space but also be affected by signals from BSs throughout the whole city. We use $\zeta_j=\|\wj-\oc\|$ to represent the distance between the city center $\oc$ and a BS located at $\wj=(\zeta_j,\gamma_j)$, where $\|\cdot\|$ means the Euclidean norm. The distance between a BS at $\wj$ and the typical IoT device at $\xx$ is 
\begin{equation}
\xi_j=\sqrt{\zeta_j^2+\zeta_0^2-2\zeta_j\zeta_0\cos(\gamma_j-\gamma_0)}.
\end{equation}
\indent As introduced in \cite{9420290}, in order to capture the imbalance of ICTs between urban and rural areas, BSs can be modeled using an inhomogeneous PPP $\Phi_{BS}\equiv \{\wj\}\subset \nbbR^2$ with a 2D-Gaussian density $\widetilde{\lambda}_p(\zeta)=\ncalG(\zeta)\lambda_p$, where 
\begin{equation}
\ncalG(\zeta)=\frac{1}{\sigma_p \sqrt{2\pi}}\exp(-\frac{\zeta^2}{2\sigma_p^2}).
\end{equation}

\subsection{Access Points Modeling}
In the considered system, we assume that {a fixed number of APs form a local network and serve the RF-powered IoT devices.} If the wireless network is built in an infinite area, we always assume that all APs follow a PPP which is a basic point process hypothesis. {But the rural area is usually finite and the funding to deploy ICTs is limited. Therefore, it is better to model the distribution of APs using a BPP \cite{haenggi2012stochastic,7882710}, where the APs are \iid in the finite region $\ncalA$ and the number of APs is $N^t$. The APs are all active and their locations are denoted as a set $\Phi_{AP}\equiv\{\yi\}\subset \nbbR^2$.} Thus, the PDF of any AP located at $\yi \in\Phi_{AP}$ is 
\begin{equation}
f(\yi)=\left\{ 
\begin{array}{cc}
\displaystyle\frac{1}{\pi r_d^2}, & \|\yi-\od\|  \leq r_d, \\
 0, & \mbox{otherwise}.
\end{array}\right.
\end{equation}

\indent For calculations of performance evaluation, we define the distance between AP at $\yi$ and the typical IoT device located at $\xx$ as $D_i$. 

\subsection{Time-slotted Architecture of IoT Devices}
\indent In this system, we assume that the RF signals are the only energy source of IoT devices. {The APs and IoT devices are {assumed to be fully-synchronized}, which can be realized through network time protocol (NTP) or mesh time-synchronization protocol (MTP) \cite{yiugitler2020overview,beke2020time}. These synchronization methods can make the clock accuracy at the millisecond (ms) level, and help the devices switch between charging and transmission modes.} All RF signals are required for communication by IoT devices in each given time slot with $T$ seconds. The energy is harvested and provided in the same time slot as introduced in \cite{6951347,8239664}. As shown in Fig.\ref{fig:Architecture}, we assume that each IoT device adopts a special architecture called the `time-slotted architecture'. In this architecture, each time slot is divided into a charging sub-slot where $T_{ch}=\tau T$ and a transmission sub-slot where $T_{tr}=(1-\tau)T$. The antenna can be used for energy charging or downlink information receiving in their sub-slots, respectively. This system can help us to evaluate the performance of energy harvesting wireless networks in rural areas using two necessary conditions described next.

\begin{figure}[ht]    
        \centering
\includegraphics[width=1\columnwidth]{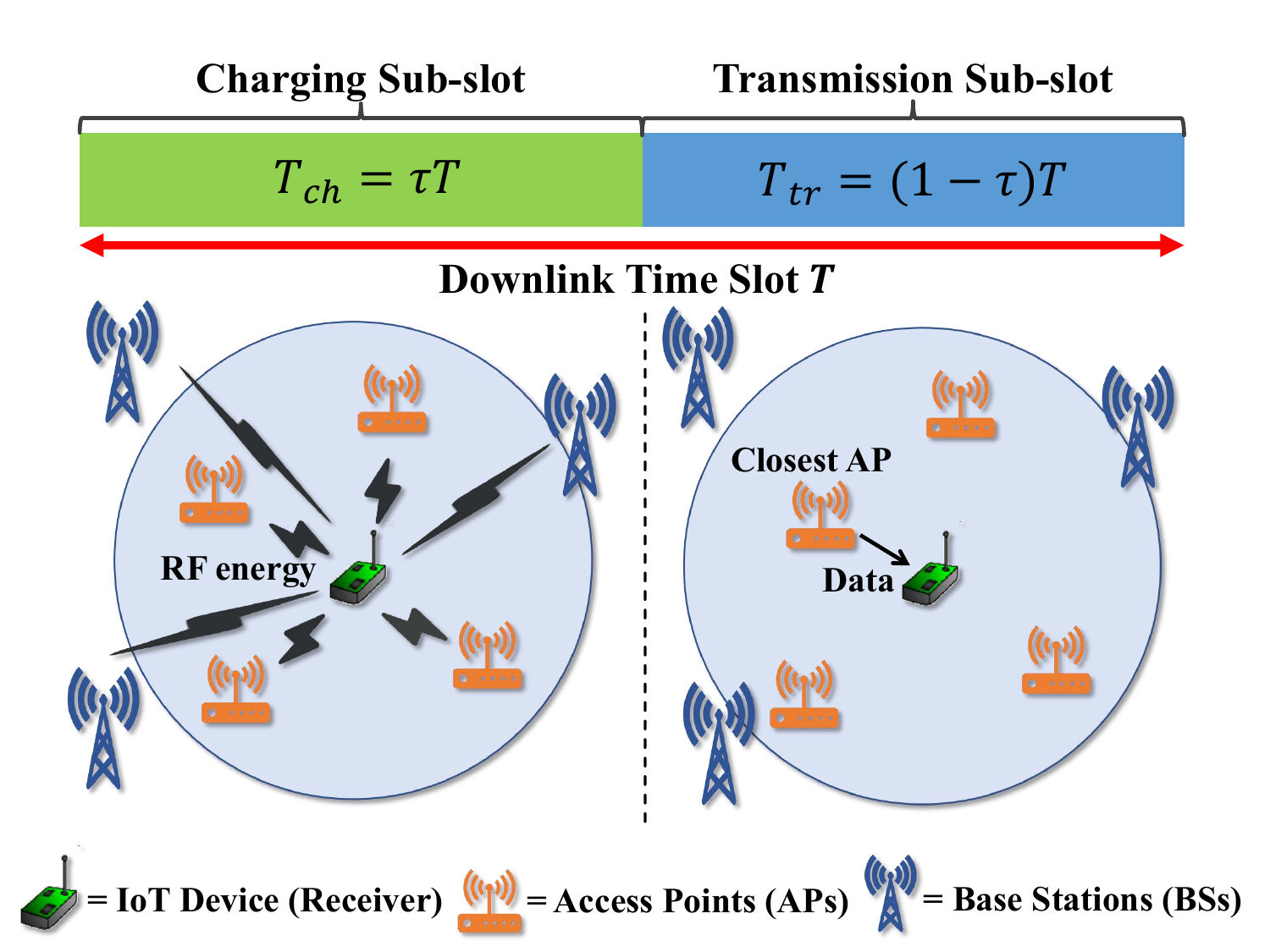}
\caption{Illustration of the IoT devices' time-slotted architecture. In the charging sub-slot, each IoT device harvests all RF signals from APs and BSs. In the transmission sub-slot, the signal from the closest AP is considered as the information source while signals from other APs and BSs are considered as interference.}
\label{fig:Architecture}
\end{figure}

\subsubsection{Charging sub-slot} In the charging sub-slot, all BSs and APs act as RF energy sources for IoT devices. The energy harvested in this sub-slot $E_H$ should be greater than the minimum energy demand $E_{\min}$ to make the transmission succeed, \ie $E_H\geq E_{\min}$, where $E_H$ can be expressed as   
\begin{equation}
    E_H=\tau T \eta (\sum_{\yi\in \Phi_{AP}}p_{AP} \gi D_i^{-\alpha}+\sum_{\wj\in \Phi_{BS}}p_{BS} \gj \xi_j^{-\alpha}),
    \label{EH2}
\end{equation}
where $\eta<1$ is the efficiency of the RF-to-DC conversion. The power of signals received by the IoT device from any of APs or BSs can be represented as $p_{AP} \gi D_i^{-\alpha}$ and $p_{BS} \gj \xi_j^{-\alpha}$, where $\gi,\gj\sim \exp(1)$ model the Rayleigh fading gains, and $D_i^{-\alpha}$, $\xi_i^{-\alpha}$ model the standard power law path-loss with exponent $\alpha>2$. We assume that the distributions and fading gains of APs and BSs are all independent when the location of the typical IoT device is fixed.\\    
\indent {Especially, when the finite area is far away from the city center, \ie $r_c$ is much larger than $\sigma_p$, we can ignore the influence from BSs. Therefore, the expression of harvested energy $E_H$ can be simplified as} 
\begin{equation}
    E_H=\tau T \eta \sum_{\yi\in \Phi_{AP}}p_{AP} \gi D_i^{-\alpha}.
    \label{EH1}
\end{equation}

\subsubsection{Transmission sub-slot} In the transmission sub-slot, we assume that the IoT devices receive the information from their associated serving AP which is chosen from all APs in the finite area $\ncalA=\textbf{b}(\od,r_d)$, where $\textbf{y}_l$ is the location of the chosen AP. Considering both the interference and thermal noise, we achieve the expression of the signal-to-interference-plus-noise ratio ({\rm SINR}) from the serving AP located at $\textbf{y}_l$ to the IoT device located at $\xx$:
\begin{equation}
    {\rm SINR}=\frac{p_{AP} H_{\textbf{y}_l}D_l^{-\alpha}}{I_l+\sigma^2}.
    \label{SINR1}
\end{equation}
\indent In (\ref{SINR1}), the interference $I_l$ is the summation of signals from APs in $\Phi_{AP} \backslash \textbf{y}_l$ and all BSs in $\Phi_{BS}$, where $\Phi_{AP} \backslash \textbf{y}_l$ contains all APs except the chosen one located at $\textbf{y}_l$. Especially, in our assumptions, APs inside the finite area form a local network. In order to obtain effective information in the local network and the largest possible {\rm SINR}, we assume that the typical IoT receives the information from the closest AP located at $\yy$ and substitute $l=1$ into (\ref{SINR1}) to represent the {\rm SINR} in this case. Therefore, the interference can be given as follows:
 \begin{equation}
     I_1=\sum\limits_{\yi \in \Phi_{AP} \backslash \textbf{y}_1}p_{AP} \hi D_i^{-\alpha}+\sum_{\wj\in \Phi_{BS}}p_{BS} \hj \xi_j^{-\alpha},
    \label{interference2}
 \end{equation}
in which $\hi, \hj\sim \exp (1)$ model the Rayleigh fading in this sub-slot and $\sigma^2$ represents the thermal noise power in each IoT device's circuit. If we want to have a successful downlink transmission, the {\rm SINR} should be larger than the threshold $\beta$, \ie ${\rm SINR}\geq \beta$. {Similar to the harvested energy, when the effect of BSs can be ignored, the interference $I_1$ in (\ref{interference2}) is simplified as shown in (\ref{interference1}):}                        
 \begin{equation}
     I_1=\sum\limits_{\yi \in \Phi_{AP} \backslash \textbf{y}_1}p_{AP} \hi D_i^{-\alpha}.
    \label{interference1}
 \end{equation}

\indent {Using the events $E_H\geq E_{\rm min}$ and ${\rm SINR}\geq \beta$, we define the coverage probabilities that can help us study the performance of IoT devices.}
\cite{yiugitler2020overview,beke2020time,gao2021has}
% \begin{definition}[Overall Coverage Probability]The probability that the energy coverage and transmission coverage are both satisfied is defined as: 
% \begin{equation}
%     P_{\rm cov}(\beta)=\E \left[\nbb1 ({\rm SINR}\geq \beta)\nbb1 (E_H\geq E_{\min}) \right].
% \end{equation}
% \end{definition}
% \indent If the communication is successful, the receiving data rate in the transmission sub-slot is $R=W\log (1+\beta)$ bps (following the Shannon Formula), where $W$ is the channel bandwidth. Using the overall coverage probability, the average downlink data rate during transmission sub-slot is
% \begin{equation}
%     R_{\avg}=\E\left[W\log(1+{\rm SINR})\nbb1 ({\rm SINR}\geq \beta)\nbb1 (E_H\geq E_{\min}) \right].
%     \label{Ravg}
% \end{equation}
% \indent Because we only use a fraction $1-\tau$ of the time slot for transmission, the average downlink data rate for the IoT receiver is 
% \begin{equation}
%     D_{avg}=(1-\tau)R_{\avg}.
% \end{equation}

\section{Distance Distribution}
In the considered system model, BSs and APs are both the RF energy sources of IoT devices. The harvested energy and interference are both related to the distance between the IoT device and BSs or APs. {Because we select the closest AP to be the serving one}, it is necessary to analyze the distribution of distances between APs and the typical IoT device.
\subsection{Distance Distribution in a Finite Area}
\indent When we analyze the distribution of APs, it is not convenient to use the distances between the rural center $\od$ and APs directly. To calculate the harvested energy and {\rm SINR}, we have introduced $D_i$ which is the distance between the typical IoT device located at $\xx$ and AP located at $\yi$. {In the considered setup, all elements in $\{D_i:D_i=\|\yi-\xx\|\}$ are correlated since they are all functions of the random variable $\xx$. However, when we condition on $\xx$, and given that the locations of APs $\{\yi\}$ are \iid, the elements in $\{D_i\}$ can be considered \iid, which helps us simplify the calculations.} Because APs follow a BPP, we can calculate the CDF of $D_i$ using the {intersection} of two circular regions: $\ncalA=\textbf{b}(\od,r_d)$ and a circular area centered at $\xx$ with radius $d_i$, which is defined as $\ncalB_i=\textbf{b}(\xx,d_i)$. The CDF of $D_i$ is derived by 
\begin{equation}
    F_{D_i}(d_i)=\P(\yi\in \ncalB_i|\yi\in \ncalA)=\frac{\E\left[\nbb1 (\yi\in \ncalA)\nbb1 (\yi\in \ncalB_i) \right]}{\E\left[\nbb1 (\yi\in \ncalA) \right]}.
\end{equation}
\indent There are two cases: \textit{i) $0\leq d_i\leq r_d-v_0$ where $\ncalB_i\subset \ncalA$}, and \textit{ii) $r_d-v_0<d_i\leq r_d+v_0$ where $\ncalB_i\backslash\ncalA\neq \varnothing$.} Based on this conditional probability, we introduce the distance distribution in a finite area in Lemma \ref{lem:RDD}. 
\begin{lemma}The CDF of the distance between AP located at $\yi$ and the typical IoT device located at $\xx$ is
\begin{equation}
    F_{D_i}(d_i)=\left\{ 
\begin{array}{cr}
F_{D_{i,1}}(d_i), & 0\leq d_i\leq r_d-v_0\\
 F_{D_{i,2}}(d_i), & r_d-v_0<d_i\leq r_d+v_0
\end{array}\right.,
\label{FFDd}
\end{equation}
with $F_{D_{i,1}}(d_i)=\frac{d_i^2}{r_d^2}$ and $F_{D_{i,2}}(d_i)=\frac{d_i^2}{\pi r_d^2}(\theta^*-\frac{1}{2}\sin{2\theta^*})+\frac{1}{\pi}(\phi^*-\frac{1}{2}\sin{2\phi^*})$, where $\theta^*=\arccos{(\frac{d_i^2+v_0^2-r_d^2}{2v_0d_i})}$ and $\phi^*=\arccos{(\frac{v_0^2+r_d^2-d_i^2}{2v_0r_d})}$.\\ 
\indent By taking the derivative of $F_{D_i}(d_i)$ and using the basic algebraic manipulations, the PDF 
$f_{D_i}(d_i)$ can be expressed as:
\begin{equation}
    f_{D_i}(d_i)=\left\{ 
\begin{array}{cr}
f_{D_{i,1}}(d_i), & 0\leq d_i\leq r_d-v_0\\
 f_{D_{i,2}}(d_i), & r_d-v_0<d_i\leq r_d+v_0
\end{array}\right.,
\label{fDd}
\end{equation}
\noindent with $f_{D_{i,1}}(d_i)=\frac{2d_i}{r_d^2}$ and $f_{D_{i,2}}(d_i)=\frac{2d_i}{\pi r_d^2}\arccos{(\frac{d_i^2+v_0^2-r_d^2}{2v_0d_i})}$.
\label{lem:RDD}
\end{lemma}
\indent Similar to the typical IoT device, APs are also constrained by spatial conditions. {Therefore, we adopt the `center-edge' distribution to model the distribution of APs and introduce the distance distribution in Corollary \ref{cor:newfDd}.}
\begin{cor}We also adopt the `center-edge distribution' to further limit the distribution of APs. The APs are distributed uniformly and randomly in these two areas: i) {\rm $\ncalA_3=\textbf{b}(\od,R_3)$}, a ring area centered at $\od$ with radius $R_3$, and ii) {\rm $\ncalA\backslash\ncalA_4=\textbf{b}(\od,r_d)\backslash\textbf{b}(\od,R_4)$}, a circular area centered at $\od$ with inner radius $R_4$ and outer radius $r_d$, where $R_4>R_3$. We use another ordered pair $(R_3, R_4)$ to represent the distribution range of APs. We use $D$ to simply represent the independent variable $D_i$. The new PDF of $D_i$ is shown as:
\begin{equation}
    f_D^{CE}(d)=\frac{r_d^2}{r_d^2-R_4^2+R_3^2}(f_D^{\ncalA}(d)-f_D^{\ncalA_4}(d)+f_D^{\ncalA_3}(d)),
\label{newfDd}
\end{equation}
and the new CDF of $D$ is:
\begin{equation}
    F_D^{CE}(d)=\frac{r_d^2}{r_d^2-R_4^2+R_3^2}(F_D^{\ncalA}(d)-F_D^{\ncalA_4}(d)+F_D^{\ncalA_3}(d)),
\label{newFFDd}
\end{equation}
where we introduce the $f_D^{\ncalA_{\ncalK}}(d)$ and $F_D^{\ncalA_{\ncalK}}(d)$ next in Lemma \ref{lem:newRDD}.
\label{cor:newfDd}
\end{cor}
After adopting the `center-edge' distribution of APs, the spatial area is divided into several parts. In Lemma \ref{lem:newRDD}, we simplify the expressions of updated PDF and CDF of $D_i$ in different domains.
\begin{lemma}Considering the positional relationship between {\rm $\ncalA_{\ncalK}=\textbf{b}(\od,R_{\ncalK})(\ncalK=3,4)$} and {\rm $\ncalB={\textbf{b}}(\xx,d)$}, we divide the entire domain of rural area into six parts shown in Fig.\ref{fig:F123456}:
\begin{equation}
    \left\{
\begin{array}{rr@{}l}
     \rmnuma_{\ncalK}:&0<\;&d<v_0-R_{\ncalK}\\
     \rmnumb_{\ncalK}:&v_0-R_{\ncalK}<\;&d<v_0+R_{\ncalK}\\
     \rmnumc_{\ncalK}:&v_0+R_{\ncalK}<\;&d<v_0+r_d\\     
\end{array}
    \right.,{\rm if}\; 0<R_{\ncalK}<v_0<r_d,
\end{equation}
and
\begin{equation}
    \left\{
\begin{array}{rr@{}l}
     \rmnumd_{\ncalK}:&0<\;&d<R_{\ncalK}-v_0\\
     \rmnume_{\ncalK}:&R_{\ncalK}-v_0<\;&d<R_{\ncalK}+v_0\\
     \rmnumf_{\ncalK}:&R_{\ncalK}+v_0<\;&d<r_d+v_0\\     
\end{array}
    \right.,{\rm if}\;0<v_0<R_{\ncalK}<r_d.
\end{equation}

\indent The $F_D^{\ncalA_{\ncalK}}(d)$ can be expressed as:
\begin{equation}
F_D^{\ncalA_{\ncalK}}(d)=\left\{
\begin{array}{cc}
     \rmnuma_{\ncalK}:&0\\
     \rmnumb_{\ncalK},\rmnume_{\ncalK}:&\ncalF(d,{R_{\ncalK}})\\
     \rmnumc_{\ncalK},\rmnumf_{\ncalK}:&\displaystyle\frac{R_{\ncalK}^2}{\pi r_d^2}\\
     \rmnumd_{\ncalK}:&\displaystyle\frac{d^2}{\pi r_d^2}\\
\end{array}
    \right.,
\end{equation}
where 
\begin{equation}
\ncalF(d,{R_{\ncalK}})=\frac{d^2}{\pi r_d^2}(\theta_{\ncalK}-\frac{1}{2}\sin 2\theta_{\ncalK})+\frac{R_{\ncalK}^2}{\pi r_d^2}(\phi_{\ncalK}-\frac{1}{2}\sin 2\phi_{\ncalK}),
\end{equation}
 and 
 \begin{equation}
 \left\{
 \begin{array}{l}
 \theta_\ncalK=\displaystyle\arccos(\frac{d^2+v_0^2-R_{\ncalK}^2}{2v_0d}) \\\phi_{\ncalK}=\displaystyle\arccos(\frac{v_0^2+R_{\ncalK}^2-d^2}{2v_0R_{\ncalK}})
 \end{array}.
 \right.
 \end{equation}
\indent The $f_D^{\ncalA_{\ncalK}}(d)$ is:
\begin{equation}
f_D^{\ncalA_{\ncalK}}(d)=\left\{
\begin{array}{cc}
     \rmnuma_{\ncalK},\rmnumc_{\ncalK},\rmnumf_{\ncalK}:&0\\
     \rmnumb_{\ncalK},\rmnume_{\ncalK}:&\displaystyle\frac{2d}{\pi r_d^2}\theta_{\ncalK}\\
     \rmnumd_{\ncalK}:&\displaystyle\frac{2d}{ r_d^2}
\end{array}.
    \right.
\end{equation}
    \label{lem:newRDD}
\end{lemma}

\begin{remark}The expressions of $F_D^{\ncalA}(d)$ and $f_D^{\ncalA}(d)$ are the same as (\ref{FFDd}) and (\ref{fDd}). Especially, $\ncalF(d,r_d)=F_{D_{i,2}}(d)$ in Lemma \ref{lem:RDD}.
\label{rem:newfDd}
\end{remark}

 \begin{figure}
    \centering
    \includegraphics[width=0.7\columnwidth]{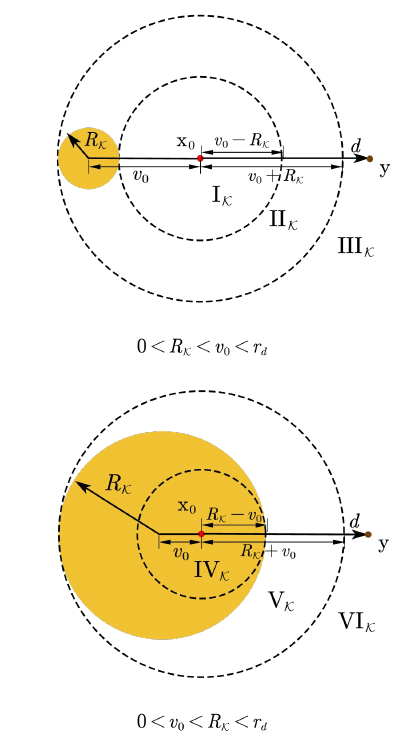}
    \caption{Illustration of distance distribution introduced in Lemma 2.}
    \label{fig:F123456}
\end{figure}
% \subsection{Closest Distance Distribution}
\subsection{Conditional Distance Distribution in a Finite Area}
\indent When we analyze the coverage performance, we define the distance between the closest AP to the typical IoT device as a random variable $R$. If the APs' distribution and location of the typical IoT device are known, $R=r=\|\xx-\yy\|$, where $\yy$ is the location of the closest AP. {In a finite area, we have introduced the BPP to model the distribution of APs,} and some studies like \cite{7882710} have introduced the distribution of the $k^{th}$ closest AP for a typical IoT device. The PDF of the distance between the $k^{th}$ closest AP and the typical IoT device can be expressed as:
\begin{equation}
    \begin{array}{@{}r@{}l}
    f_{R }^{(k)}(r)&=\frac{N^t!}{(k-1)!(N^t-k)!}F_{D_{i }}(r)^{k-1}f_{D_{i }}(r)(1-F_{D_{i }}(r))^{N^t-k}.%\displaystyle
    \end{array}
\label{fkRr}
\end{equation}
%\\&\ \ \ \ \ \ \ \ \ \ \ \ \ \ \ \ \ \ \ \ \ \ \ \ \ \ \ \ \times 
\indent To achieve the distribution of $R$, we substitute $k=1$ into (\ref{fkRr}) and achieve the distribution of the distance between the closest AP and the typical IoT device in Lemma \ref{lem:CDD}.

\begin{lemma} The PDF of the distance between the typical IoT device located at $\xx$ and its closest AP located at $\yy$ is:
\begin{equation}
f_{R }(r)=N^t f_{D_{i }}(r)(1-F_{D_{i }}(r))^{N^t-1},
\label{fRr}
\end{equation}
where $f_{D_i}(d_i)$ and $F_{D_i}(d_i)$ are given in (\ref{fDd}) and (\ref{FFDd}), respectively.
\label{lem:CDD}
\end{lemma}
% \begin{IEEEproof}
%     The closest distance distribution is the result after substituting $k=1$ into (\ref{fkRr}).
% \end{IEEEproof}
\begin{figure}[ht]    
        \centering
\includegraphics[width=1\columnwidth]{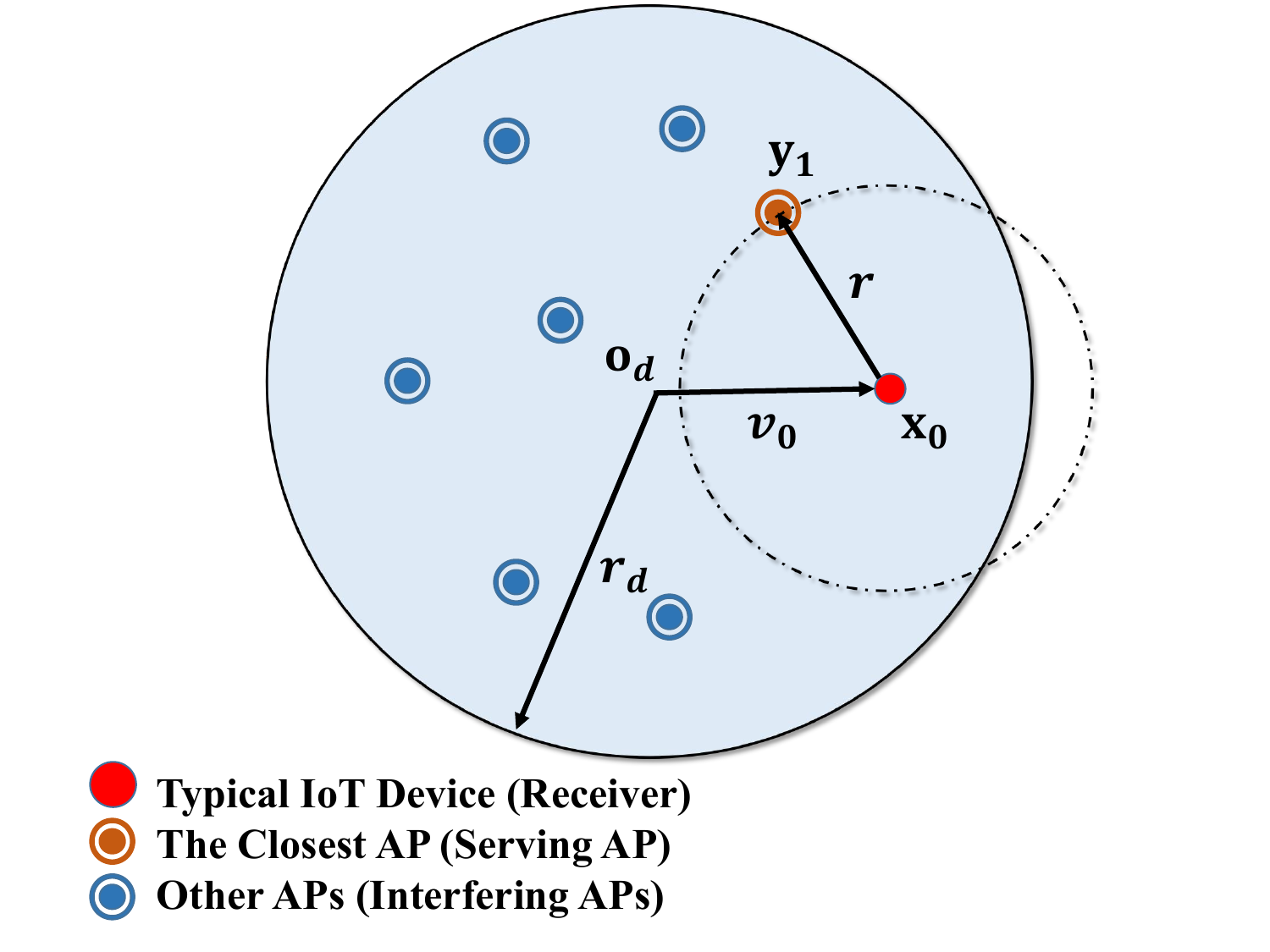}
\caption{Illustration of the BPP model and the distance distribution. Because we assume that the typical IoT device can receive the signal from the closest AP, other APs are regarded as interfering APs. Lemma \ref{lem:CDD} introduces the distribution of the distance between the typical IoT device and the closest AP, and Lemma \ref{lem:CRDD} introduces the distribution of $D_i$ conditioned on $R=r$.}
\label{fig:BPP}
\end{figure}

\indent We also update the PDF of $R$ when adopting the `center-edge' distribution in Corollary \ref{cor:newfRr}.
\begin{cor}When we adopt the assumption in Corollary \ref{cor:newfDd}, the new PDF of the distance between the typical IoT device located at $\xx$ and its closest AP located at $\yy$ is
\begin{equation}
    f_R^{CE}(r)=N^t f_D^{CE}(r)(1-F_D^{CE}(r))^{N^t-1},
\label{newfRr}
\end{equation}
where the $f_D^{CE}(d)$ and $F_D^{CE}(d)$ are shown in (\ref{newfDd}) and (\ref{newFFDd}).
\label{cor:newfRr}
\end{cor}

\indent To simplify the calculation in coverage probability analysis, we need to define a new distribution of $D_i$ conditioned on $R=r$ using a new variant $U$. The conditional PDF of $D_i$ is expressed as $f_U(u)$ in Lemma \ref{lem:CRDD}.
\begin{lemma}[Conditional Distance Distribution in a Finite Area]The PDF of $D_i$ conditioned on $R=r$ is
\begin{equation}
    \begin{array}{@{}r@{}l}
    f_U&(u)=\displaystyle\frac{f_{D_i(u)}}{1-F_{D_i}(r)}\\&=\left\{
    \begin{array}{rrr}
     \displaystyle\frac{f_{D_{i,1}}(u)}{1-F_{D_{i,1}}(r)},& 0<r<d^-,& r<u<d^-\\
     \displaystyle\frac{f_{D_{i,2}}(u)}{1-F_{D_{i,1}}(r)},& 0<r<d^-,& d^-<u<d^+\\
     \displaystyle\frac{f_{D_{i,2}}(u)}{1-F_{D_{i,2}}(r)},& d^-<r<d^+,& r<u<d^+
    \end{array}
    \right.,
    \end{array}
    \label{fUu}
\end{equation}
where $d^-=r_d-v_0$ and $d^+=r_d+v_0$. $F_{D_i}(d)$ and $f_{D_i}(d)$ are given in (\ref{FFDd}) and (\ref{fDd}).
\label{lem:CRDD}
\end{lemma}
\begin{IEEEproof}
Because the condition is $R=r$, we can use the Bayesian Probability Formula to achieve the conditional probability easily. {Similar to the concept in \cite{7882710}, we consider the situation that all APs are outside the circle centered at $\xx$ with radius $r$  because we choose the closest AP for downlink transmission.}
\end{IEEEproof}
Similar to Corollary \ref{cor:newfDd} and Corollary \ref{cor:newfRr}, the Corollary \ref{cor:newfUu} describes the expression of new PDF of $D_i$ conditioned on $R=r$.  
\begin{cor}
    Following the assumptions in Corollary \ref{cor:newfDd}, \ref{cor:newfRr} and Lemma \ref{lem:CRDD}, the APs are distributed following the `center-edge distribution'. In this case, $f_U(u)$ can be expressed as:
\begin{equation}
    f_U^{CE}(u)=\frac{f_D^{CE}(d)}{1-F_D^{CE}(r)},
\label{newfUu}
\end{equation}
where $f_D^{CE}(d)$ and $F_D^{CE}(d)$ are given in (\ref{newfDd}) and (\ref{newFFDd}).
\label{cor:newfUu}
\end{cor}

\section{Coverage Probability Analysis}
\indent In our considered system, we assume that APs follow a BPP $\Phi_{AP}$ and BSs follow an inhomogeneous PPP $\Phi_{BS}$ with a 2D-Gaussian density. There are two necessary conditions to achieve device coverage: energy supply and SINR requirement. In this section, we first define the energy coverage probability (ECP) and transmission coverage probability (TCP). After that, we derive the expression for the overall coverage probability (OCP) $P_{\rm cov}$ using the conditional energy and transmission coverage probabilities.

\begin{table*}[t]\caption{Table of Notations}
\centering
\begin{center}
\resizebox{\textwidth}{!}{
\renewcommand{\arraystretch}{1}%1.4
    \begin{tabular}{ {c} | {l} }
    \hline
        \hline
    \textbf{Notation} & \textbf{Description} \\ \hline
    $\Phi_{AP}$; $\yi$; $N^t$ & BPP modeling the locations of APs; location of any AP; the number of all APs\\ \hline
    $\Phi_{BS}$; $\wj$ & Inhomogeneous PPP modeling the locations of BSs; location of any BS\\ \hline
    $\lambda_p$; $\sigma_p$ & Parameters in the Gaussian density of the inhomogeneous PPP $\Phi_{BS}$\\ \hline    
    $\xx$; $\yy$ & Location of the typical IoT device; location of the closest AP to the typical IoT device\\ \hline
    $\od$; $\oc$& Location of the center of the finite area; location of the city center\\ \hline    
    $r_d$; $r_c$ & Radius of the finite area; distance between $\od$ and $\oc$\\ \hline    
    $\gi$; $\gj$ & Rayleigh fading gains in the charging sub-slot\\ \hline
     $\hi$; $\hj$ & Rayleigh fading gains in the transmission sub-slot\\ \hline
    $\p_{AP}$; $\p_{BS}$ & Uniform power of APs; uniform power of BSs\\ \hline    
    $\tau$; $T$; $\eta$ & Time slot division parameter for charging sub-slot; length of each time slot; efficiency of the RF-to-DC conversion\\ \hline    
    $\sigma^2$; $\alpha$; $\beta$ & Power of thermal noise; path loss exponent ($\alpha>2$); {\rm SINR} threshold for a successful transmission\\ \hline 
    % $P_{\rm cov}$; $W$ & The overall coverage probability; the bandwidth of the downlink channel\\ \hline 
    % $R_{\avg}$; $D_{\avg}$ & The average downlink data rate during transmission; the average downlink data rate during the whole time-slot\\ \hline 
    $R_1$; $R_2$ & Parameters used to limit the distribution of IoT devices in rural areas\\ \hline 
    $R_3$; $R_4$ & Parameters used to limit the distribution of APs in rural areas\\ \hline 

     \hline
    \end{tabular}
    }
\end{center}
\label{tab:TableOfNotations}
%\vspace{-8mm}
\end{table*}

\subsection{Energy Coverage Probability}
\indent To activate its circuit, the typical IoT device needs to harvest enough energy during the charging sub-slot. In order to analyze the energy harvesting performance of IoT devices, we introduce the energy coverage probability in Theorem \ref{theo:ECP}.
\begin{theorem}[Energy Coverage Probability, ECP] Conditioned on the location of the typical IoT device and $R=r$, the probability of the event that $E_H\geq E_{\min}$ is: %the energy harvested during the charging sub-slot $E_H$ is greater than the minimum demand $E_{\min}$ is:
\begin{equation}
\begin{array}{r@{}l}
    \P&(E_H\geq E_{\min}|r,v_0,\psi) \\&=\displaystyle\exp\Big(-r^{\alpha}[\ncalC(\tau)-\Omega_{AP}(r)-\frac{p_{BS}}{p_{AP}}\Omega_{BS}]^+\Big),
\end{array}
\end{equation}
where 
$$\Omega_{AP}(r)=\displaystyle(N^t-1)\int_{r}^{r_d+v_0}u^{-\alpha}f_U(u)\dd u,$$
and
$$\Omega_{BS}=\displaystyle\int_0^{2\pi}\int_0^{\infty}\widetilde{\lambda}_p(\zeta)\xi(\zeta,\gamma)^{-\alpha}\zeta \dd \zeta \dd\gamma.$$ 
\indent In addition, $\ncalC(\tau)=\frac{E_{\min}}{\tau T \eta p_{AP}}$ and $f_{U}(u)$ is shown in (\ref{fUu}). The expression of energy coverage probability is the expectation of conditional ECP over $r$. 
\begin{equation}
\begin{array}{@{}l}
    \P(E_H\geq E_{\min}|v_0,\psi)=\displaystyle\E_{R}[\P(E_H\geq E_{\min}|R,v_0,\psi)]\\
    \ =\displaystyle\E_R\Big[\exp(-R^{\alpha}[\ncalC(\tau)-\Omega_{AP}(R)-\frac{p_{BS}}{p_{AP}}\Omega_{BS}]^+)\Big]\\
    \ =\displaystyle\int_0^{\ncalS}f_{R}(r)\dd r \\\ + \displaystyle\int_{\ncalS}^{r_d+v_0}\exp\big(-r^{\alpha}[\ncalC(\tau)-\Omega_{AP}(r)-\displaystyle\frac{p_{BS}}{p_{AP}}\Omega_{BS}]\big)f_{R}(r)\dd r.
\end{array}
\end{equation}
where $f_{R}(r)$ is shown in (\ref{fRr}), $\ncalS$ is the threshold satisfying $\displaystyle\ncalC(\tau)-\Omega_{AP}(\ncalS)-\frac{p_{BS}}{p_{AP}}\Omega_{BS}=0$.
\label{theo:ECP}
\end{theorem}
\begin{IEEEproof}
See Appendix~\ref{app:ECP}.
\end{IEEEproof}
{When the remoteness of the finite area increases, the BS density decreases and such a finite area can be used to model a rural area. For most rural areas far from the city center, it is an acceptable assumption to ignore the effect of BSs. In this case, the energy coverage probability can be simplified as shown in Corollary \ref{cor:ECP1}.}
\begin{cor}
If the RF signals from BSs can be ignored, the conditional ECP can be written as:
\begin{equation}
    \P(E_H\geq E_{\min}|r,v_0)=\exp(-r^{\alpha}[\ncalC(\tau)-\Omega_{AP}(r)]^+),
\end{equation}
and the ECP can be written as:
\begin{equation}
\begin{array}{@{}r@{}l}
    \P(E_H\geq E_{\min}|v_0)&=\E_{R}[\P(E_H\geq E_{\min}|R,v_0)].
\end{array}
\end{equation}

\label{cor:ECP1}
\end{cor}
% \begin{IEEEproof}
% See the Appendix~\ref{app:ECP}.
% \end{IEEEproof}

\begin{remark}
    When the location of the typical IoT device is fixed, the energy harvesting performance can be analyzed using $\P(E_H\geq E_{\min}|v_0)$. But when it is uniformly and randomly distributed in the rural area, the energy harvesting performance should be evaluated using the expectation $\E_V[\P(E_H\geq E_{\min}|V)]$. 
    \label{rem:ECP}
\end{remark}
\subsection{Transmission Coverage Probability}

\indent Conditioned on the distance between the typical IoT device and its closest AP $R=r$, the transmission coverage probability can be expressed using the Laplace transform of the interference of APs and BSs as shown in Lemma \ref{lem:TCP}.
\begin{lemma}[Transmission Coverage Probability, TCP]
Conditioned on $R=r$, the probability of the event ${\rm SINR}\geq\beta$ is:
\begin{equation}
\begin{array}{r@{}l}
    \displaystyle\P&({\rm SINR}\geq \beta|r,v_0,\psi) \\&=\exp(-\frac{\beta r^{\alpha} \sigma^2}{p_{AP}})\ncalL_{AP}(\beta r^{\alpha})\ncalL_{BS}\bigg(\frac{p_{BS}}{p_{AP}}\beta r^{\alpha}\bigg),
\end{array}
\end{equation}
where 
\begin{equation}
\begin{array}{@{}l}
\ncalL_{BS}(s)=\exp(-\displaystyle\int_0^{2\pi}\int_0^{\infty} \widetilde{\lambda}(\zeta)(1-\frac{1}{1+s\xi(\zeta,\gamma)^{-\alpha}})\zeta \dd\zeta \dd\gamma)
\end{array}
\end{equation}
and 
\begin{equation}
\ncalL_{AP}(s)=\bigg(\int_0^{r_d+v_0}\frac{1}{1+s u^{-\alpha}}f_U(u)\dd u\bigg)^{N^t-1},
\end{equation}
where $f_U(u)$ is given in (\ref{fUu}) and $f_{R}(r)$ is given in (\ref{fRr}). 
% The transmission coverage probability is its expectation about the closest relevant distance $r$.
% \begin{equation}
%     \begin{array}{@{}r@{}l}
%     \P({\rm SINR}\geq \beta|v_0)&=\E_{R}[\P({\rm SINR}\geq \beta|R,v_0)]\\
%     &=\int_0^{r_d+v_0}\exp(-\frac{\beta r^{\alpha} \sigma^2}{p_{AP}})\ncalL_{AP}(\beta r^{\alpha})f_{R}(r)\dd r,
%     \end{array}
% \end{equation}

\label{lem:TCP}
\end{lemma}
\begin{IEEEproof}
See Appendix~\ref{app:TCP}.
\end{IEEEproof}
Transmission coverage probability is a vital tool to calculate the overall coverage probability. We also simplify the expression of TCP in Corollary \ref{cor:TCP1} when the interference of BSs can be ignored.
\begin{cor}
If the effect of BSs can be ignored, the interference can be expressed as shown in (\ref{interference1}). In this situation, the conditional TCP can be written as:
\begin{equation}
    \P({\rm SINR}\geq \beta|r,v_0)=\exp(-\frac{\beta r^{\alpha} \sigma^2}{p_{AP}})\ncalL_{AP}(\beta r^{\alpha}).
\end{equation}
%     The expression of transmission coverage probability is the expectation of conditional TCP over $r$.
% \begin{equation}
%     \P({\rm SINR}\geq \beta|v_0)=\E_R\bigg[\exp(-\frac{\beta R^{\alpha} \sigma^2}{p_{AP}})\ncalL_{AP}(\beta R^{\alpha})\ncalL_{BS}(\frac{p_{BS}}{p_{AP}}\beta R^{\alpha})\bigg]  
% \end{equation}
\label{cor:TCP1}
\end{cor}
% \begin{IEEEproof}
% See the Appendix~\ref{app:TCP}.
% \end{IEEEproof}

\subsection{Overall Coverage Probability}
We have obtained the expressions of ECP and TCP conditioned on $R=r$, but our final goal is to achieve the expression of overall coverage probability. Therefore, we introduce the overall coverage probability in Theorem \ref{theo:OCP}.
\begin{theorem}[Overall Coverage Probability, OCP] The overall coverage probability can be derived using the conditional energy coverage probability and transmission coverage probability. It can be expressed as:
\begin{equation}
    \begin{array}{@{}r@{}l}
    P&_{\rm cov}(\beta|v_0,\psi)\\&=\displaystyle\int_0^{r_d+v_0}\exp(-r^{\alpha}[\ncalC(\tau)-\Omega_{AP}(r)-\frac{p_{BS}}{p_{AP}}\Omega_{BS}]^+)\\ 
    &\times \displaystyle\exp(-\frac{\beta r^{\alpha} \sigma^2}{p_{AP}})\ncalL_{AP}(\beta r^{\alpha})\ncalL_{BS}\bigg(\frac{p_{BS}}{p_{AP}}\beta r^{\alpha}\bigg)f_{R}(r)\dd r,
    \end{array}
\end{equation}
\textsl{where $\ncalC(\tau)=\frac{E_{\min}}{\tau T \eta p_{AP}}$, $[x]^+=\max\{0,x\}$. The $\Omega_{AP}(r)$, $\Omega_{BS}$ and $\ncalL_{AP}(s)$, $\ncalL_{BS}(s)$ are defined in Theorem \ref{theo:ECP} and Lemma \ref{lem:TCP}, while $f_{R}(r)$ and $f_U(u)$ are given in (\ref{fRr}) and (\ref{fUu})}.
\label{theo:OCP}
\end{theorem}
\begin{IEEEproof}
See Appendix~\ref{app:OCP}.
\end{IEEEproof}
\indent As the combination of ECP and TCP, the expression of overall coverage probability can be simplified when ignoring RF signals from BSs in Corollary \ref{cor:OCP1}.
\begin{cor}
If the RF signals from BSs can be ignored, the overall coverage probability can be written as:
\begin{equation}
    \begin{array}{@{}r@{}l}
    P_{\rm cov}(\beta|v_0)&=\displaystyle\int_0^{r_d+v_0}\exp(-r^{\alpha}[\ncalC(\tau)-\Omega_{AP}(r)]^+)\\
    &\displaystyle\times\exp(-\frac{\beta r^{\alpha} \sigma^2}{p_{AP}}) \ncalL_{AP}(\beta r^{\alpha})f_{R}(r)\dd r.
    \end{array}
\end{equation}
\label{cor:OCP1}
\end{cor}
% \begin{IEEEproof}
% See the Appendix~\ref{app:OCP}.
% \end{IEEEproof}

\begin{remark}
    When the location of the typical IoT device is fixed, the overall performance can be analyzed using $P_{\rm cov}(\beta|v_0)$. But when the IoT devices are uniformly and randomly distributed inside the finite area, the overall coverage performance should be evaluated using $\E_V[P_{\rm cov}(\beta|V)]$. 
    \label{rem:OCP}
\end{remark}

% \begin{remark}
%     We notice that the number of APs $N^t$ has two conflicting effects on the overall coverage probability. When $N^t$ increases, the IoT device's energy demand is supplied more easily in the charging sub-slot, but the device is also affected by more interfering APs in the transmission sub-slot. So when $N^t$ is small, the OCP is mainly affected by the ECP because the interference is also small. But when $N^t$ is large, the highest value of OCP is mainly affected by TCP due to the sufficient energy supply.
% \end{remark}

% \indent Based on the overall coverage probability, we can work out the average data rate mentioned in (\ref{Ravg}).

% \begin{lemma}[Average Downlink Data Rate]The average data rate during the downlink transmission sub-slot for the IoT device is:
% \begin{equation}
% \begin{array}{@{}r@{}l}
%     R_{\avg}=W\log&(1+\beta)P_{\rm cov}(\beta|v_0)+\int_{W\log(1+\beta)}^{\infty}P_{\rm cov}(2^{\frac{z}{W}}-1|v_0)\dd z
% \end{array}
% \end{equation}
% \label{lem:ADDR}
% \end{lemma}
% \begin{IEEEproof}
% See the Appendix \ref{app:Ravg}. 
% \end{IEEEproof}

\section{Simulation Results and Discussion}
We first consider both the signals from APs and BSs, {and prove that the remoteness of the finite area can affect the energy harvesting performance of IoT devices inside the finite wireless network.} We set the radius of the finite area as $r_c=2\,{\rm km}$ and there are $2\times 10^4 $ APs in it. The parameters of the Gaussian BS density are set as $\lambda_p=2\,{\rm BSs/km^2}$ and $\sigma_p=2\,{\rm km}$. Other parameters in the simulation are set as follows: $\alpha=3$, $\tau=0.2$, $T=1\,{\rm s}$, $p_{AP}=1\,{\rm W}$, $p_{BS}=10\,{\rm W}(1\,{\rm dBW})$, $\eta=0.5$, and $E_{\min}=10^{-4}\,{\rm J}$. {The parameters we choose have been verified in the literature \cite{8536464,liu2019energy,9420290,ye2021joint,8239664}, respectively. As shown in Fig.\ref{fig:remoteness}, if $r_c$ is much less than $\sigma_p$, the IoT devices inside the finite area have a uniform and high energy supply. When $r_c$ is close to $\sigma_p$, the IoT devices closer to the city center can have better energy coverage than those farther from the city center. As the finite area moves away from the city center, the energy coverage probability decreases, When $r_c$ is much larger than $\sigma_p$, the IoT devices also have a uniform energy supply but lower than in the previous case. Such a finite area can be used to model a communication network in rural area.} 
\begin{figure*}[htbp]
    \centering
    \includegraphics[width=1.6\columnwidth]{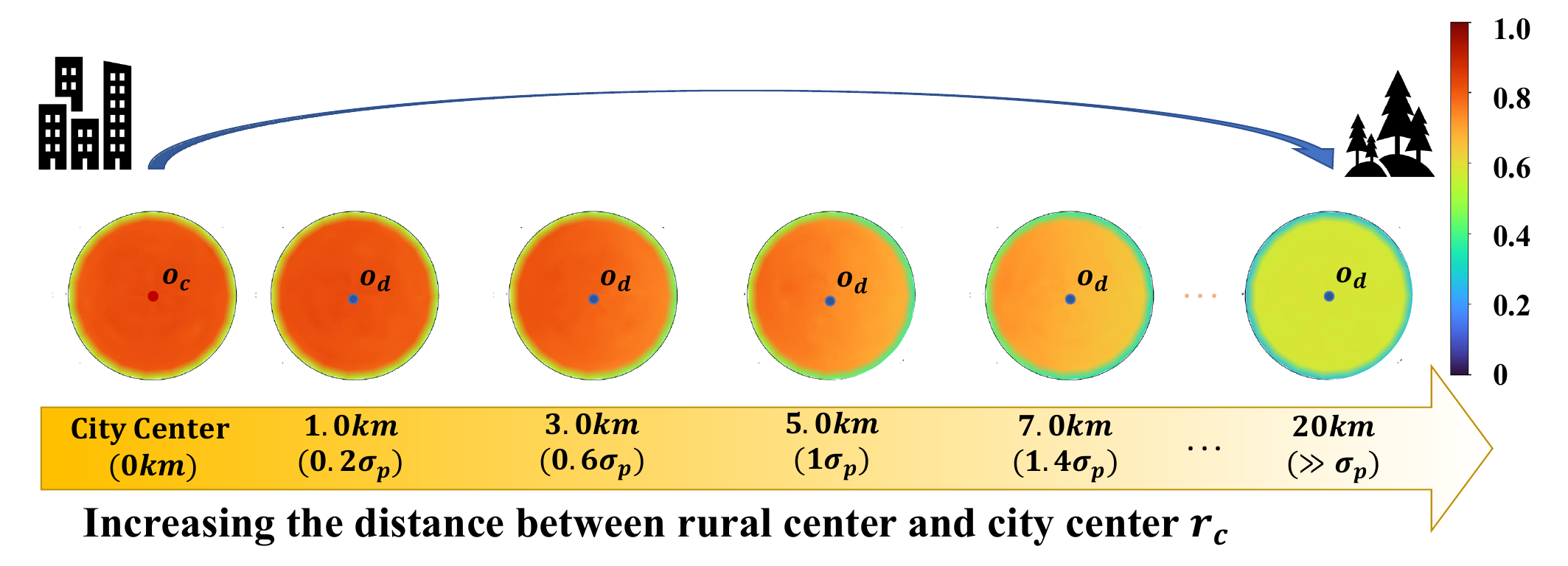}
    \caption{The figure shows the effect of remoteness of finite areas on energy harvesting performance. When $r_c=0\,{\rm km}$, the energy harvesting performance is the highest. When $r_c$ is close to $\sigma_p$, the IoT devices closer to the city center have better energy harvesting performance than the devices further away from the city center. When $r_c \gg\sigma_p$, the energy harvesting performance is the lowest.}
    \label{fig:remoteness}
\end{figure*}

\indent {Firstly, we analyze the performance of a typical IoT device inside a rural area. We consider a small rural area and reset the radius of rural area $r_d=100\,{\rm m}$, while other parameters about transmission are set as: $\beta=0.1(-10\,{\rm dB})$ and $\sigma^2=10^{-9}$. We simulate the energy coverage probability and overall coverage probability to observe their relation to $v_0$ and $N^t$.} Fig.\ref{fig:ECP_v0} and Fig.\ref{fig:OCP_v0} show the theoretical analysis and simulation results of ECP and OCP when we choose different values of $N^t$. {When the typical IoT device is located inside the finite space ($v_0<80\,{\rm m}$), ECP and OCP are both generally stable. This shows that APs beyond a certain distance only have a small impact on IoT devices. When the AP density in the area near IoT devices is uniform, their performance becomes more stable. But when it is close to the edge of the finite space, the ECP is lower. Both ECP and OCP generally decrease when the distance between the IoT device and the rural center increases, which is because the IoT device near the edge can not harvest enough energy to support its operation.} %And more details are shown in Fig.\ref{fig:ECPsim_Nt_v0} and Fig.\ref{fig:OCPsim_Nt_v0}.\\
\begin{figure}[htbp]
\subfigure[Energy coverage probability derived in Theorem \ref{theo:ECP}.]{
\begin{minipage}[t]{1\linewidth}
\centering 
\includegraphics[width=1\textwidth]{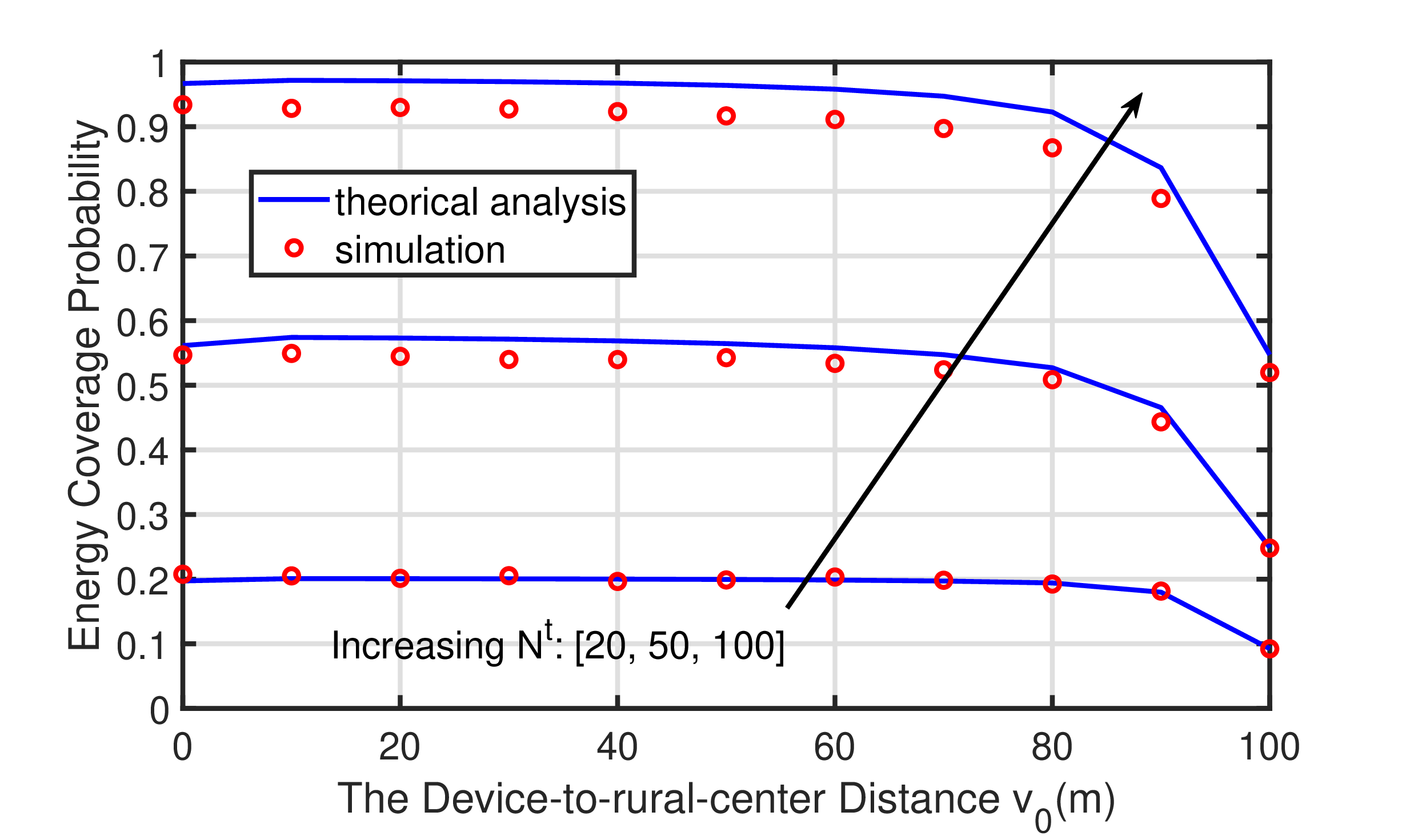}
\label{fig:ECP_v0}
\end{minipage}}
\subfigure[Overall coverage probability derived in Theorem \ref{theo:OCP}.]{
\begin{minipage}[t]{1\linewidth}
\centering 
\includegraphics[width=1\textwidth]{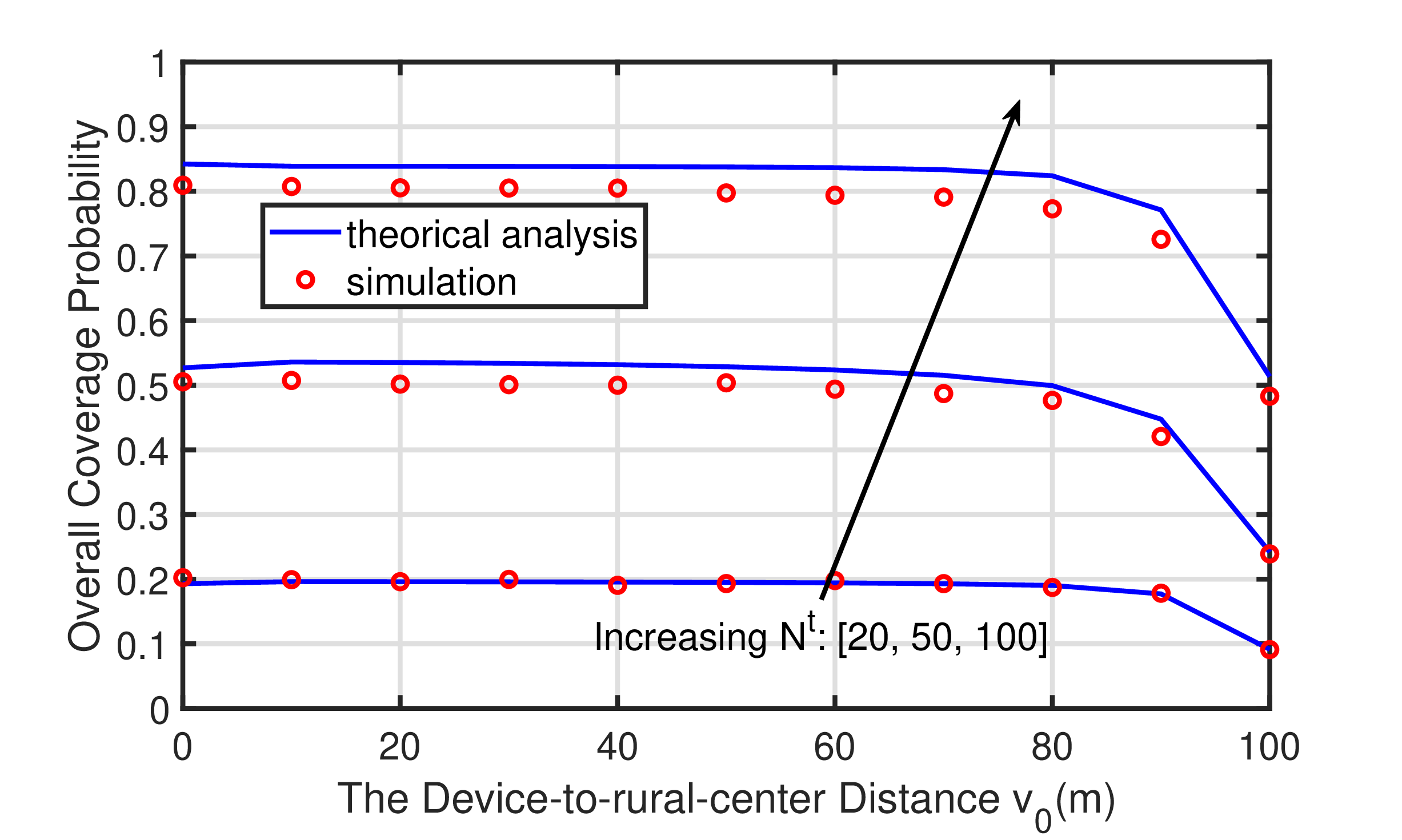}
\label{fig:OCP_v0}
\end{minipage}}
\caption{Theoretical analysis and simulation results of ECP and OCP conditioned on $v_0$.}
\end{figure}

% \begin{figure}[ht]    
%         \centering
% \includegraphics[width=0.9\columnwidth]{figures/ECP_v0.png}
% \caption{The theoretical analysis and simulation results of energy coverage probability derived in Theorem \ref{theo:ECP}. } 
% \label{fig:ECP_v0}
% \end{figure}

% \begin{figure}[ht]    
%         \centering
% \includegraphics[width=0.9\columnwidth]{figures/OCP_v0.png}
% \caption{The theoretical analysis and simulation results of overall coverage probability derived in Theorem \ref{theo:OCP}. }
% \label{fig:OCP_v0}
% \end{figure}

% \begin{figure}[ht]    
%         \centering
% \includegraphics[width=1\columnwidth]{figures/ECP_Nt.png}
% \caption{Theoretical analysis of energy coverage probability derived in Theorem \ref{theo:ECP} compared to the simulation results.}
% \label{fig:ECP_Nt}
% \end{figure}
% \begin{figure}[ht]    
%         \centering
% \includegraphics[width=1\columnwidth]{figures/OCP_Nt.png}
% \caption{Theoretical analysis of overall coverage probability derived in Theorem \ref{theo:OCP} compared to the simulation results.}
% \label{fig:OCP_Nt}
% \end{figure}
{Secondly, we evaluate the average performance of IoT devices inside a rural area.} In Fig.\ref{fig:ECP_Nt} and Fig.\ref{fig:OCP_Nt}, we plot the theoretical analysis curves of the expectation of ECP and OCP over $v_0$, which are mentioned in Remark \ref{rem:ECP} and Remark \ref{rem:OCP}. We compare them to the simulation results. When we set up enough APs to provide energy to IoT devices, the ECP approaches 1. When $N^t$ is small, the OCP has a similar trend to the ECP, but it does not approach 1 when $N^t$ is large due to high interference. Since the number of APs $N^t$ has different effects on ECP and TCP, when we use OCP as the evaluation standard, we can choose the optimal $N^t$ to achieve the trade-off. For example, when $r_d=20\,{\rm m}$, the optimal number of APs is 10. But when $r_d=50\,{\rm m}$, the optimal $N^t$ is 60.
\begin{figure}[htbp]
\subfigure[The expectation of energy coverage probability over $v_0$ mentioned in Remark \ref{rem:ECP}.]{
\begin{minipage}[t]{1\linewidth}
\centering 
\includegraphics[width=1\textwidth]{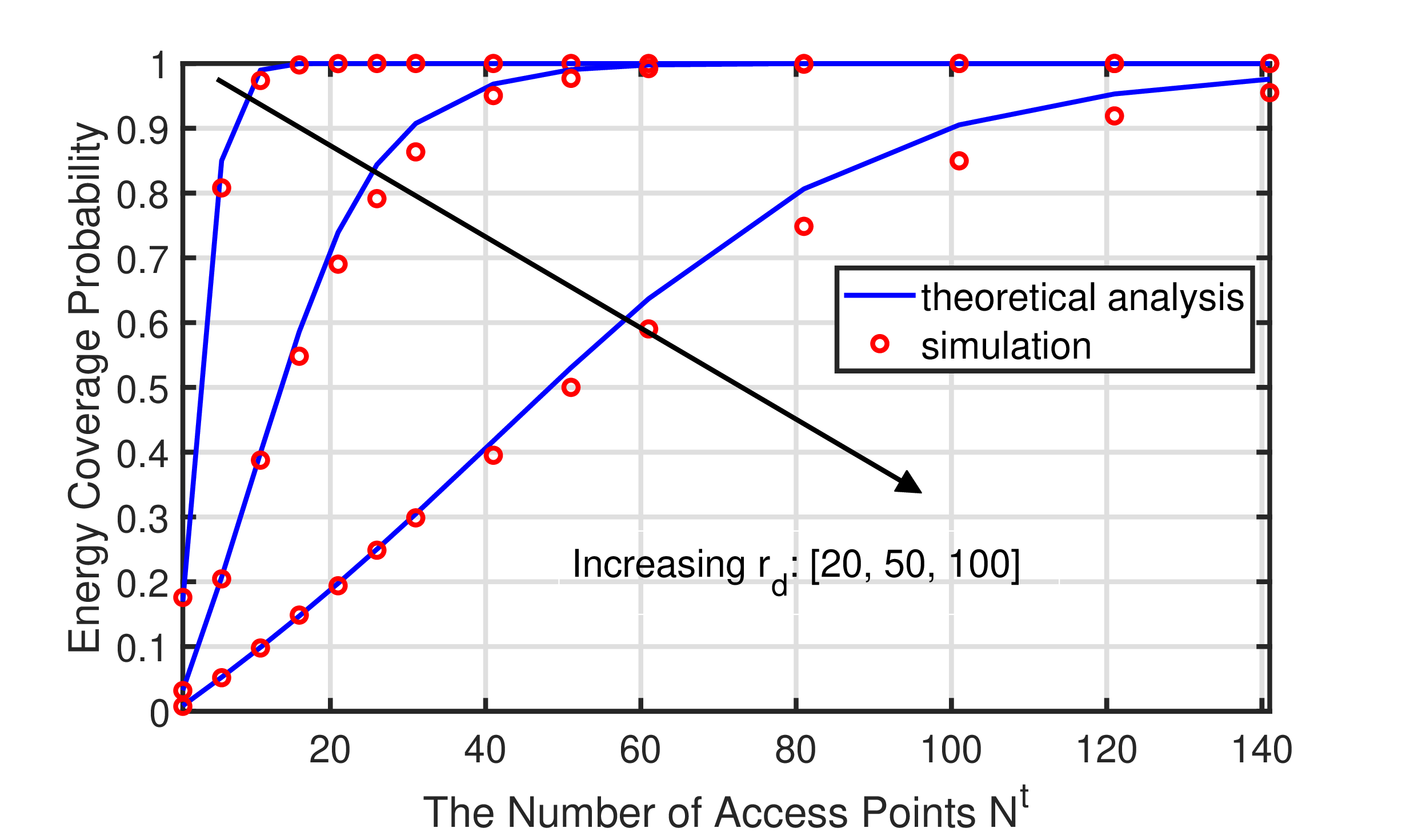}
\label{fig:ECP_Nt}
\end{minipage}}
\subfigure[The expectation of overall coverage probability over $v_0$ mentioned in Remark \ref{rem:OCP}.]{
\begin{minipage}[t]{1\linewidth}
\centering 
\includegraphics[width=1\textwidth]{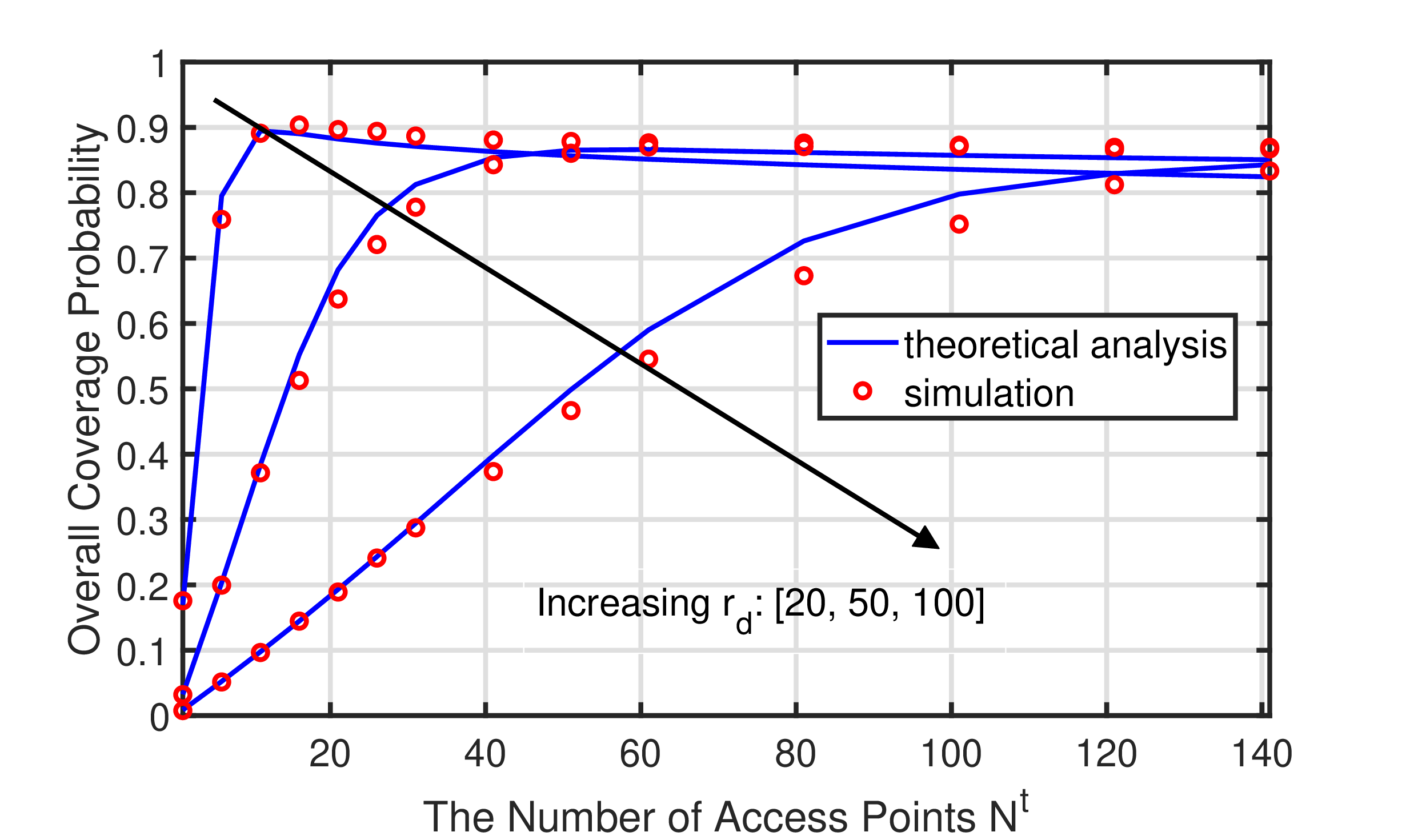}
\label{fig:OCP_Nt}
\end{minipage}}
\caption{Theoretical analysis of ECP and OCP compared to the simulation results conditioned on $N^t$.}
\end{figure}
% \begin{figure}[ht]    
%         \centering
% \includegraphics[width=0.9\columnwidth]{figures/ECPsim_Nt_v0.png}
% \caption{The energy coverage probability heatmap (simulation)}
% \label{fig:ECPsim_Nt_v0}
% \end{figure}
% As shown in Fig.\ref{fig:ECPsim_Nt_v0}, when increasing the number of the APs $N^t$, the ECP is higher, which means the energy demand in the IoT device is supplied more easily. Similarly Fig.\ref{fig:OCPsim_Nt_v0} is a heatmap which shows the simulation result of the OCP which has a similar trend to Fig.\ref{fig:ECPsim_Nt_v0}. But the OCP cannot reach the same highest value as ECP, because it's highest value is constrained by the transmission performance. 

% \begin{figure}[ht]    
%         \centering
% \includegraphics[width=0.9\columnwidth]{figures/OCPsim_Nt_v0.png}
% \caption{The overall coverage probability heatmap (simulation)}
% \label{fig:OCPsim_Nt_v0}
% \end{figure}

%  \begin{figure*}[htbp]
% \subfigure[The heatmap of overall coverage probability (simulation) when $(R_1,R_2)=(30,60)$ and $N^t=10$.]{
% \begin{minipage}[t]{0.5\linewidth}
% \centering 
% \includegraphics[width=1\textwidth]{figures/OCPsim_30_60.png}
% \label{fig:OCPsim_3060}
% \end{minipage}}
% \subfigure[The theoretical analysis and simulation result of the optimal edge of APs' distribution when $(R_1,R_2)=(30,60)$ and $N^t=10$.]{
% \begin{minipage}[t]{0.5\linewidth}
% \centering 
% \includegraphics[width=1\textwidth]{figures/Optimal.pdf}
% \label{fig:Optimal_3060}
% \end{minipage}}
% \caption{Theoretical analysis and simulation result when $(R_1,R_2)=(30,60)$ and $N^t=10$.}
% \end{figure*}
\begin{figure}[ht]    
        \centering
\includegraphics[width=1\columnwidth]{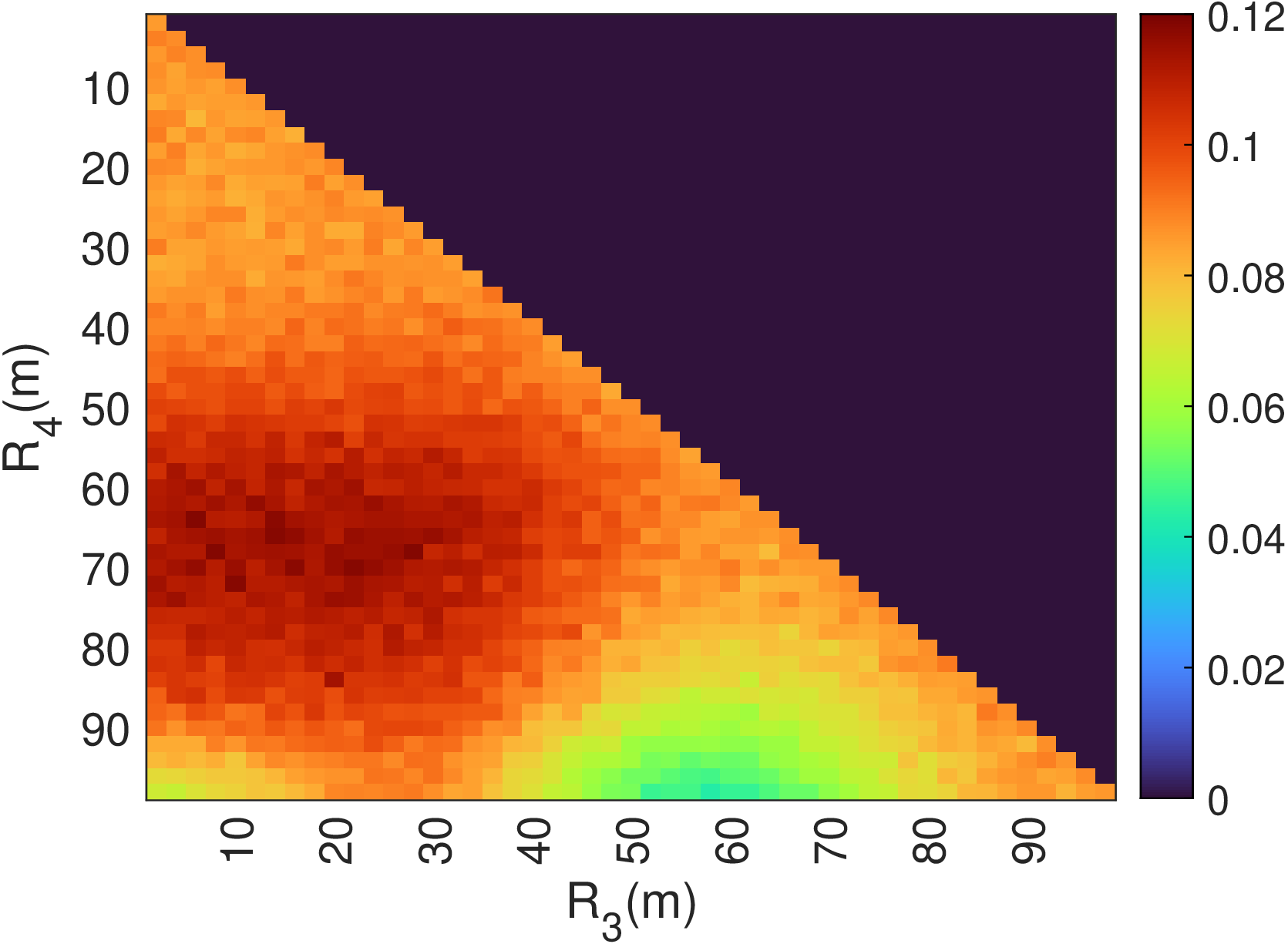}
\caption{The heatmap of overall coverage probability (simulation) when $(R_1,R_2)=(30,60)$ and $N^t=10$. }
\label{fig:OCPsim_3060}
\end{figure}
\begin{figure}[ht]    
        \centering
\includegraphics[width=1\columnwidth]{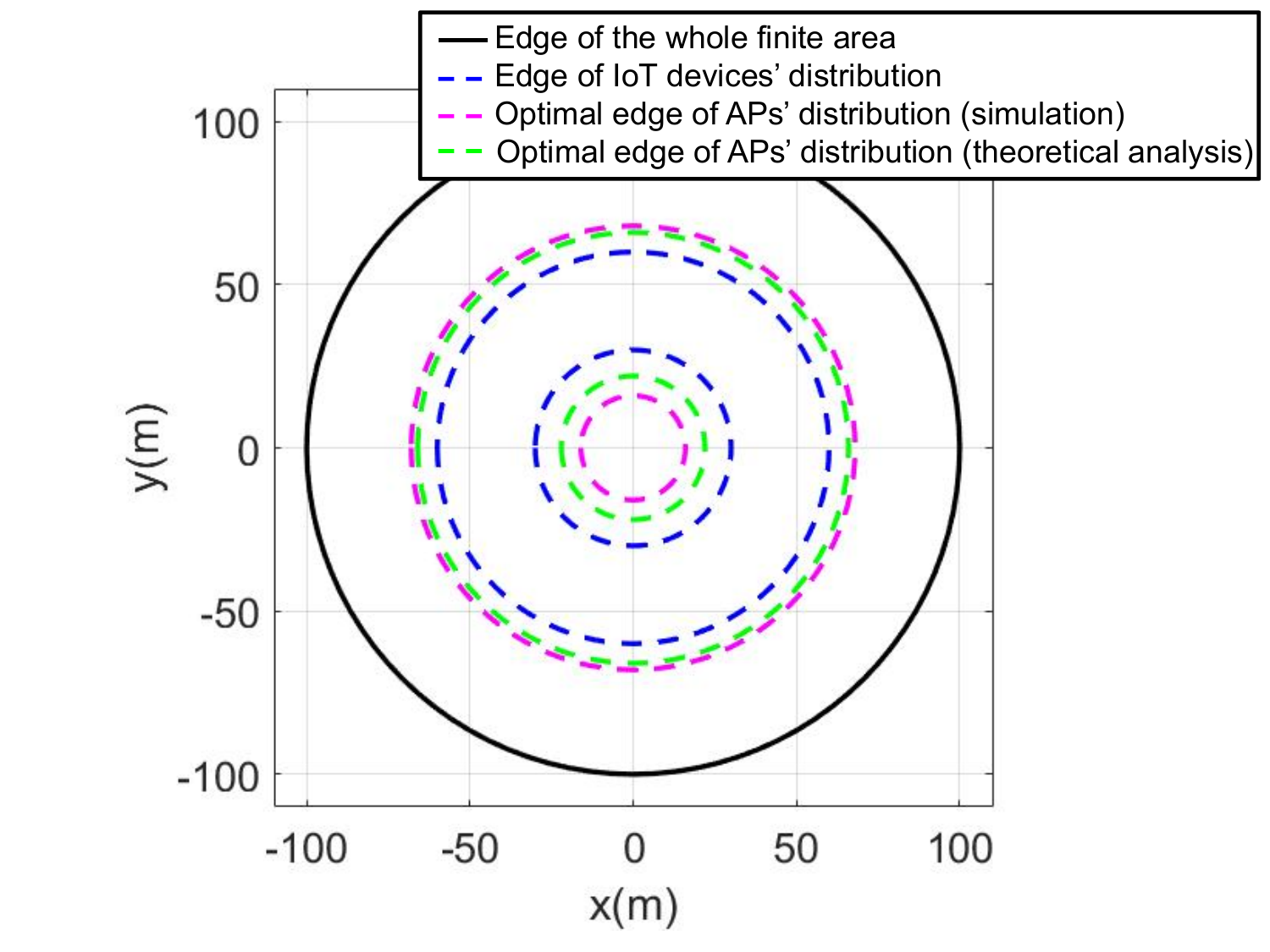}
\caption{The theoretical analysis and simulation result of the optimal edge of APs' distribution when $(R_1,R_2)=(30,60)$ and $N^t=10$.}
\label{fig:Optimal_3060}
\end{figure}

As introduced in Corollary \ref{cor:newfDd}, we adopt the `center-edge distribution' when we limit the distribution area of IoT devices and APs. When we change the distribution range of the IoT devices, the optimal location of $(R_3, R_4)$ is changed at the same time. Fig.\ref{fig:OCPsim_3060} is an example when $(R_1, R_2)=(30,60)$. Fig.\ref{fig:Optimal_3060} shows the theoretical analysis and simulation result of the optimal edge of APs' distribution in this case. The optimal value of $R_3$ is less than $R_1$, and the optimal value of $R_4$ is larger than $R_2$. That means the optimal distribution range of APs satisfies $\ncalA_3\subset \ncalA_1$ and $\ncalA\backslash\ncalA_4\subset \ncalA\backslash\ncalA_2$ in our simulation.\\
% \begin{figure}[htbp]
% \subfigure[The heatmap of overall coverage probability (simulation).]{
% \begin{minipage}[t]{1\linewidth}
% \centering 
% \includegraphics[width=1\textwidth]{figures/OCPsim_30_60.png}
% \label{fig:OCPsim_3060}
% \end{minipage}}
% \subfigure[Optimal edges of APs' distribution.]{
% \begin{minipage}[t]{1\linewidth}
% \centering 
% \includegraphics[width=1\textwidth]{figures/Optimal.pdf}
% \label{fig:Optimal_3060}
% \end{minipage}}
% \caption{Theoretical analysis and simulation result when $(R_1,R_2)=(30,60)$ and $N^t=10$.}
% \end{figure}
\indent In Fig.\ref{fig:R1234}, the grid color corresponding to the better ordered pair $(R_3, R_4)$ is closer to dark red. When we fix the value of $R_1$ and increase the value of $R_2$, the optimal $R_3$ is not greatly affected, but the value of the optimal $R_4$ increases. Similarly, the optimal $R_4$ is larger when $R_2$ increases. 

\begin{figure}[ht]    
        \centering
\includegraphics[width=1\columnwidth]{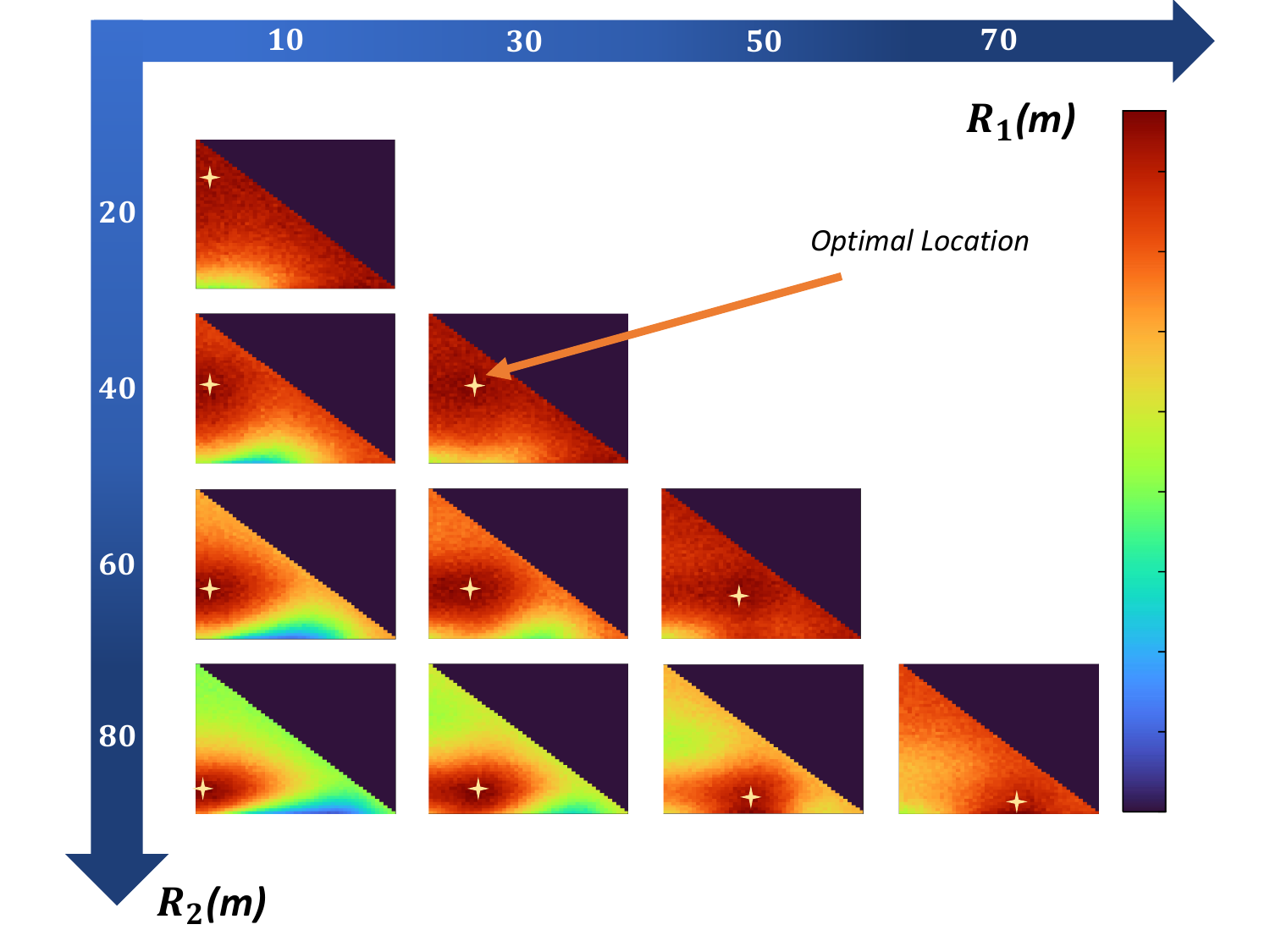}
\caption{The optimal choice for $(R_3, R_4)$ when $(R_1, R_2)$ have different choices and $N^t=10$ (simulation). }
\label{fig:R1234}
\end{figure}

{In the considered system, all APs and IoT devices are synchronized. However, in realistic applications, the synchronization issues may result in loss of wireless energy, signal power, and interference, which complicates the performance analysis. Therefore, synchronization challenges and their impact on rural wireless networks are expected to be further analyzed.} {In addition, the real-world constraints of IoT devices and APs are complex, which makes the performance analysis more difficult. Electromagnetic environment simulation platforms can help analyze the coverage performance. However, the cost of building a virtual environment for each different rural area is high. Therefore, although our proposed system is an idealistic model, it is still a tractable and convenient solution to analyze the coverage analysis of RF-powered IoT networks in rural areas before deployment.}

\section{Conclusion}
{
In this paper, we build a mathematical model to analyze the downlink performance of RF-powered IoT devices in rural areas. We use an inhomogeneous PPP to describe the difference in ICT density between the urban area and rural area in a large-scale city. At the same time, we use a BPP to model the distribution of a fixed number of APs inside a considered finite area. We focus on the performance of battery-less IoT devices inside the finite area, which harvest and process RF signals in different sub-slots. The energy coverage, transmission coverage, and overall coverage probability are defined to analyze the performance of IoT devices. \\
\indent We show that due to the effect of BSs in the whole city, IoT devices closer to the city center have higher energy harvesting performance than those far away from the city center. This means that unbalanced BS density causes the wireless power digital divide. Because rural areas can be modeled as a finite area far from the city center, the effect (including energy supply and interference) from BSs can be ignored. We prove that the IoT devices closer to the center of the finite area can work better than those towards the edge. We propose the `center-edge distribution' to model the deployment constraints of IoT devices and APs in rural areas. Based on this, we prove that an optimal APs' distribution exists when the IoT devices' distribution is constrained.} \\
\appendices
\section{Proof of Theorem~\ref{theo:ECP}}\label{app:ECP}
\indent In the considered system, we assume that $r=\|\xx-\yy\|\in(0,r_d+v_0)$, which is the distance between the closest AP and the typical IoT device. Fig.\ref{fig:rural} shows the illustration of BSs' distribution. Using the expression for harvested energy $E_H$ shown in (\ref{EH2}), we can derive the conditional energy coverage probability:
\begin{equation}
    \begin{array}{@{}r@{}l}
    &\P(E_H\geq E_{\min}|r,v_0,\psi)\\
    &\ =\P\bigg(\tau T \eta (\displaystyle\sum\limits_{\yi\in \Phi_{AP}}p_{AP} \gi D_i^{-\alpha}\\&\ \;\;\;\;\;\;\;\;\;\;\;\;\;\;\;\;\;\;\;\;+\displaystyle\sum\limits_{\wj\in \Phi_{BS}}p_{BS} \gj \xi_j^{-\alpha})\geq E_{\min}\bigg|r,v_0,\psi\bigg) \\
    &\ = \P\bigg(\g1 D_1^{-\alpha}+ \displaystyle\sum\limits_{\yi\in \Phi_{AP} \backslash \yy }\gi D_i^{-\alpha}\\&\ \;\;\;\;\;\;\;\;\;\;\;\;\;\;\;\;\;\;\;\;\;\;+\displaystyle\frac{p_{BS}}{p_{AP}}\displaystyle\sum_{\wj\in \Phi_{BS}}\gj \xi_j^{-\alpha} \geq \ncalC(\tau)\bigg|r,v_0,\psi\bigg)\\
    &\ \overset{(a)}{\approx} \P\bigg(\g1 r^{-\alpha}+\Omega_{AP}(r)+\displaystyle\frac{p_{BS}}{p_{AP}}\Omega_{BS}\geq \ncalC(\tau)\bigg|r,v_0,\psi\bigg),
    \end{array}
\end{equation}
where $\ncalC(\tau)=\frac{E_{\min}}{\tau T \eta p_{AP}}$. In step (a), we define $\Omega_{AP}(r)$ as a conditional expectation of fading gains of APs except the closest one, and it can be derived as:
\begin{equation}
    \begin{aligned}
    \Omega_{AP}(r)&=\E\left[\sum\limits_{\yi\in\Phi_{AP} \backslash \yy }\gi D_i^{-\alpha}\Big|r,v_0,\psi\right]\\
    &\overset{(b)}{=}\E\left[\sum\limits_{\yi\in\Phi_{AP}\backslash \yy }D_i^{-\alpha}|r,v_0 \right]\\
    &\overset{(c)}{=}(N^t-1)\int_{r}^{r_d+v_0}\frac{1}{u^{\alpha}}f_U(u)\dd u,
    \end{aligned}
\end{equation}
where step (b) comes from $\gi\sim\exp(1)$, and step (c) follows the assumption that the APs are \iid in our system. Because BSs follow an inhomogeneous PPP with a 2D-Gaussian density $\widetilde{\lambda}_p(\zeta)$ and $\gj\sim\exp(1)$ models the Rayleigh fading, $\Omega_{BS}$ can be calculated as:
\begin{equation}
\begin{array}{@{}r@{}l}
    \Omega_{BS}&=\displaystyle\E\Big(\sum_{\wj\in \Phi_{BS}}\gj \xi_j^{-\alpha}|r,v_0,\psi\Big)\\
    &=\displaystyle\E\Big(\sum_{\wj\in \Phi_{BS}}\xi_j^{-\alpha}|v_0,\psi\Big)\\
    &=\displaystyle\int_0^{2\pi}\int_0^{\infty}\widetilde{\lambda}_p(\zeta)\xi(\zeta,\gamma)^{-\alpha}\zeta \dd \zeta \dd\gamma.
    \end{array}
\end{equation}
\indent Because $\g1 \sim \exp(1)$, the expression of the energy coverage probability can be derived as follows:
\begin{equation}
    \begin{array}{@{}r@{}l}
    &\P(E_H\geq E_{\min}|r,v_0)\\
    &\ =\P\bigg(\g1 \geq r^{\alpha}\Big[\ncalC(\tau)-\Omega_{AP}(r)-\displaystyle\frac{p_{BS}}{p_{AP}}\Omega_{BS}\Big]\bigg|r,v_0\bigg)\\
    &\ \overset{(d)}{=}\exp\bigg(-r^{\alpha}\Big[\ncalC(\tau)-\Omega_{AP}(r)-\displaystyle\frac{p_{BS}}{p_{AP}}\Omega_{BS}\Big]^+\bigg),
    \end{array}
    \label{ECP_AP}
\end{equation}
where step (d) comes from the CCDF of exponential distribution of $\g1$, and $[x]^+=\max\{0,x\}$ in this step. Since $\Omega_{AP}(r)$ is a monotonically decreasing and continuous function of $r$ where $\Omega_{AP}(r) \in (0,+\infty)$, there is only one solution $r=\ncalS$ for the equation $\ncalC(\tau)-\Omega_{AP}(r)-\Omega_{BS}=0$. Thus, $[\ncalC(\tau)-\Omega_{AP}(r)-\frac{p_{BS}}{p_{AP}}\Omega_{BS}]^+$ can be written as:
\begin{equation}
\begin{array}{@{}l}
    \displaystyle\Big[\ncalC(\tau)-\Omega_{AP}(r)-\frac{p_{BS}}{p_{AP}}\Omega_{BS}\Big]^+\\
    \ =\left\{
    \begin{array}{cl}
        \ncalC(\tau)-\displaystyle\Omega_{AP}(r)-\frac{p_{BS}}{p_{AP}}\Omega_{BS}, &\ncalS<r<r_d+v_0\\
        0, &0<r\leq \ncalS
    \end{array}
    \right..
\end{array}
\end{equation}%\\    \ 
\indent We can calculate the expectation of energy coverage probability using $f_{R}(r)$: 
\begin{equation}\
    \begin{array}{@{}r@{}l}
    &\P(E_H\geq  E_{\min}|v_0,\psi) =\E_{R}\bigg[\P(E_H\geq E_{\min}|R,v_0,\psi)\bigg]\\
    &\ =\displaystyle\E_{R}\bigg[\exp\bigg(-R^{\alpha}\Big[\ncalC(\tau)-\Omega_{AP}(R)-\frac{p_{BS}}{p_{AP}}\Omega_{BS}\Big]^+\bigg)\bigg]\\
    &\ =\displaystyle\int_0^{r_d+v_0} \exp(-r^{\alpha}\Big[\ncalC(\tau)-\Omega_{AP}(r)-\frac{p_{BS}}{p_{AP}}\Omega_{BS}\Big]^+)f_{R}(r)\dd r\\
    &\ =\displaystyle\int_0^{\ncalS} f_{R}(r)\dd r \\
    &\ +\displaystyle\int_{\ncalS}^{r_d+v_0} \exp\big(-r^{\alpha}\Big[\ncalC(\tau)-\Omega_{AP}(r)-\displaystyle\frac{p_{BS}}{p_{AP}}\Omega_{BS}\Big]\big)f_{R}(r)\dd r.
    \end{array}
\end{equation}
\indent  If the rural area is far from the city center, the effect of BSs can be ignored, which means the summation of BSs' fading gains $\Omega_{BS}$ approach 0. In this case, the conditional ECP can be simplified as $\P(E_H\geq E_{\min}|r,v_0)$ using the harvested energy in (\ref{EH1}):
\begin{equation}
    \begin{array}{@{}r@{}l}
    &\P(E_H\geq E_{\min}|r,v_0)\\
    &\ =\P\bigg(\tau T \eta \sum\limits_{\yi\in\Phi_{AP}}p_{AP} \gi D_i^{-\alpha}\geq E_{\min}\bigg) \\
    &\ = \displaystyle\P\bigg(\g1 D_1^{-\alpha}+ \sum\limits_{\yi\in \Phi_{AP} \backslash \yy }\gi D_i^{-\alpha}\geq \ncalC(\tau)\bigg)\\
    &\ \overset{(e)}{\approx} \P\bigg(\g1 r^{-\alpha}+\Omega_{AP}(r)\geq \ncalC(\tau)\bigg),
    \end{array}
\end{equation}
where step (e) is similar to step (a). In this case, the conditional ECP is:

% in which
% \begin{equation}
% \begin{array}{@{}r@{}l}
%     \displaystyle\iint_P\cdot\zeta \dd\zeta \dd\gamma=\big(\displaystyle\int_0^{y^*}+\displaystyle\int_{2\pi-y^*}^{2\pi}\big)\big(\displaystyle\int_0^{\zeta_M(\gamma)}
%     +\displaystyle\int_{\zeta_X(\gamma)}^{\infty}\big)\cdot\zeta \dd\zeta\dd\gamma+\displaystyle\int_{y^*}^{2\pi-y^*}\displaystyle\int_0^{\infty}\cdot\zeta \dd\zeta \dd\gamma.
% \end{array}
% \end{equation}
%, where $\gamma^*=\arcsin (\frac{r_d}{r_c})$, $\zeta_M(\gamma)=r_c\cos\gamma-\sqrt{r_d^2-r_c^2\sin^2\gamma}$ and $\zeta_X(\gamma)=r_c\cos\gamma+\sqrt{r_d^2-r_c^2\sin^2\gamma}$.

\begin{equation}
    \P(E_H\geq E_{\min}|r,v_0)=\exp\Big(-r^{\alpha}\Big[\ncalC(\tau)-\Omega_{AP}(r)\Big]^+\Big).
    \label{ECP_AP_BS}
\end{equation}

\section{Proof of Lemma~\ref{lem:TCP}}\label{app:TCP}
The overall coverage probability, based on the distribution of the closest distance to APs, is conditioned on the harvested energy and {\rm SINR}. To analyze it, we need to derive the expression of the transmission coverage probability conditioned on $R=r$, which is written as $\P({\rm SINR}\geq \beta|r,v_0)$. With this condition, the power of the signal from the closest AP received by the typical IoT device is $p_{AP} \h1 r^{-\alpha}$, and the summation of those from other APs and BSs can be written as (\ref{interference2}). The conditional transmission coverage probability can be derived as follows:
\begin{equation}
    \begin{array}{@{}l@{}l}
    \P\left({\rm SINR}\geq \beta|r,v_0,\psi\right)
=\P\left(\displaystyle\frac{p_{AP} \h1 r^{-\alpha}}{I_1+\sigma^2}\geq \beta|r,v_0,\psi\right)\\
    \ =\P\bigg(\h1\geq \displaystyle\frac{\beta r^{\alpha}\sigma^2}{p_{AP}}+\beta r^{\alpha}\displaystyle\sum \limits_{\yi \in \Phi_{AP} \backslash \yy}\hi D_i^{-\alpha}\\
    \ \;\;\;\;\;\;\;\;\;\;\;\;\;\;\;\;\;\;\;\;\;\;\;\;\;\;\;+\displaystyle\frac{p_{BS}}{p_{AP}}\beta r^{\alpha}\displaystyle\sum \limits_{\wj \in \Phi_{BS}}\hj \xi_j^{-\alpha}\bigg|r,v_0,\psi\bigg)\\
    \ =\exp(-\displaystyle\frac{\beta r^{\alpha} \sigma^2}{p_{AP}})\\ \ \times\E\bigg[\exp\Big(-\beta r^\alpha \displaystyle\sum \limits_{\yi \in \Phi_{AP} \backslash \yy}\hi D_i^{-\alpha}\Big)\bigg|r,v_0,\psi\bigg]\\
    \ \times\E \bigg[\exp\Big(-\displaystyle\frac{p_{BS}}{p_{AP}}\beta r^\alpha \displaystyle\sum \limits_{\wj \in \Phi_{BS}}\hj \xi_i^{-\alpha}\Big)\bigg|r,v_0,\psi\bigg].
    % \\
    % =\exp(-\frac{\beta r^{\alpha} \sigma^2}{p_{AP}})\ncalL_{AP}(\beta r^{\alpha})\ncalL_{BS}\Big(\displaystyle\frac{p_{BS}}{p_{AP}}\beta r^{\alpha}\Big).
    \end{array}
\end{equation}

\indent We define the Laplace transform of fading gains from other APs conditioned on $R=r$ as:

\begin{equation}
    \begin{array}{@{}r@{}l}
    &\ncalL_{AP}(s)=\E\bigg[\exp\bigg(-s \displaystyle\sum \limits_{\yi \in \Phi_{AP} \backslash \yy}\hi D_i^{-\alpha}\bigg)\bigg|r,v_0,\psi\bigg]\\
    &\ =\E_H\bigg[ \displaystyle\prod \limits_{\yi \in \Phi_{AP} \backslash \yy} \exp\bigg(-s \hi D_i^{-\alpha} \bigg)\bigg|r,v_0\bigg]\\
    &\ \overset{(f)}{=}\E\bigg[ \displaystyle\prod \limits_{\yi \in \Phi_{AP} \backslash \yy} \displaystyle\frac{1}{1+s D_i^{-\alpha}}\bigg|r,v_0\bigg]\\
    &\ \overset{(g)}{=}\bigg(\displaystyle\int_0^{r_d+v_0}\displaystyle\frac{1}{1+su^{-\alpha}}f_U(u)\dd u\bigg)^{N^t-1},
    \end{array}
\end{equation}
where step (f) comes from $\hi\sim \exp(1)$, and step (g) follows from the conditional PDF of $D_i$. Differently, the Laplace transform of relative fading gains from BSs is:
\begin{equation}
\begin{array}{@{}r@{}l}
    &\ncalL_{BS}(s)=\E\bigg[\exp\bigg(-s \displaystyle\sum \limits_{\wj \in \Phi_{BS}}\hj \xi_i^{-\alpha}\bigg)\bigg|r,v_0,\psi\bigg]\\
    &\ =\E_H\bigg[\displaystyle\prod\limits_{\wj \in \Phi_{BS}}\exp\bigg(-s \hj \xi_i^{-\alpha}\bigg)\bigg|v_0,\psi\bigg]\\
    &\ =\E\bigg[\displaystyle\prod\limits_{\wj \in \Phi_{BS}}\exp\bigg(\frac{1}{1+s \xi_i^{-\alpha}}\bigg)\bigg|v_0,\psi\bigg]\\
    &\ \overset{(h)}{=}\exp(-\displaystyle\int_0^{2\pi}\int_0^{\infty} \widetilde{\lambda}(\zeta)\bigg(1-\displaystyle\frac{1}{1+s\xi(\zeta,\gamma)^{-\alpha}}\bigg)\zeta d\zeta d\gamma),
\end{array}
\end{equation}
where step (h) is based on the PGFL of the inhomogeneous PPP.
Thus, the conditional TCP is
\begin{equation}
\begin{array}{r@{}l}
    \P&({\rm SINR}\geq \beta|r,v_0,\psi)\\ \ &=\displaystyle\exp(-\frac{\beta r^{\alpha} \sigma^2}{p_{AP}})\ncalL_{AP}(\beta r^{\alpha})\ncalL_{BS}\Big(\displaystyle\frac{p_{BS}}{p_{AP}}\beta r^{\alpha}\Big).
\end{array}
\end{equation}
% \indent We can use the conditional TCP and the closest distance distribution to derive the expectation of TCP:
% \begin{equation}
%     \begin{array}{@{}r@{}l}
%     \P({\rm SINR}\geq \beta|v_0)&=\E_{R}\left[\P({\rm SINR}\geq \beta|R,v_0)\right]\\
%     &=\displaystyle\int_0^{r_d+v_0}\P({\rm SINR}\geq \beta|r,v_0)f_{R}(r)\dd r\\
%    &=\displaystyle\int_0^{r_d+v_0}\exp(-\displaystyle\frac{\beta r^{\alpha} \sigma^2}{p_{AP}})\ncalL_{AP}(\beta r^{\alpha})\ncalL_{BS}\Big(\displaystyle\frac{p_{BS}}{p_{AP}}\beta r^{\alpha}\Big)f_{R}(r)\dd r,
%     \end{array}
% \end{equation}
% where the $f_{R}(r)$ is the PDF of the closest distance, which is shown in (\ref{fRr}).\\
\indent When the BSs' interference can be ignored, the interference can be simplified as (\ref{interference1}), and the conditional TCP is simplified as follows:

\begin{equation}
    \begin{array}{@{}r@{}l}
    &\P\left({\rm SINR}\geq \beta|r,v_0\right) =\P\left(\displaystyle\frac{p_{AP} \h1 r^{-\alpha}}{I_1+\sigma^2}\geq \beta\Big|r,v_0,\psi\right)\\
    &\ =\P\left(\h1\geq \displaystyle\frac{\beta r^{\alpha}\sigma^2}{p_{AP}}+\beta r^{\alpha}\sum \limits_{\yi \in \Phi_{AP} \backslash \yy} \hi D_i^{-\alpha}\Big|r,v_0\right)\\
    &\ =\E\left[\exp(-\displaystyle\frac{\beta r^{\alpha}\sigma^2}{p_{AP}}-\beta r^{\alpha}\sum \limits_{\yi \in \Phi_{AP} \backslash \yy} \hi D_i^{-\alpha})\bigg|r,v_0\right]\\
    &\ =\exp(-\displaystyle\frac{\beta r^{\alpha} \sigma^2}{p_{AP}})\E \bigg[\exp\bigg(-\beta r^\alpha \displaystyle\sum \limits_{\yi \in \Phi_{AP} \backslash \yy}\hi D_i^{-\alpha}\bigg)\bigg|r,v_0\bigg]\\
    &\ =\exp(-\displaystyle\frac{\beta r^{\alpha} \sigma^2}{p_{AP}})\ncalL_{AP}(\beta r^{\alpha})
    .
    \end{array}
\end{equation}

\section{Proof of Theorem~\ref{theo:OCP}}\label{app:OCP}
Considering the two main system requirements $E_H$ and {\rm SINR}, the overall coverage probability can be derived as
\begin{equation}
    \begin{array}{@{}r@{}l}
    &P_{\rm cov}(\beta|v_0,\psi)=\E[\nbb1 ({\rm SINR}\geq \beta)\nbb1 (E_H\geq E_{\min})]\\
    &\ \overset{(i)}{=}\E_{R}\bigg[\E_H[\nbb1 ({\rm SINR}\geq \beta)|R,v_0,\psi]\\
    &\ \;\;\;\;\;\;\;\;\;\;\;\;\;\;\;\;\;\;\;\;\;\;\;\;\;\;\;\;\;\;\;\;\;\;\;\;\;\;\;\times\E_G[\nbb1 (E_H\geq E_{\min})|R,v_0,\psi]\bigg]\\
    &\ =\E_{R}\bigg[\P({\rm SINR}\geq \beta|R,v_0,\psi)\times \P(E_H\geq E_{\min}|R,v_0,\psi)\bigg],
    \label{Pcovb}
    \end{array}
\end{equation}
where step (i) is the result of the assumption that fading gains during two sub-slots are independent. The two conditional probability has been proved in Theorem \ref{theo:ECP} and Lemma \ref{lem:TCP}. Substituting these two conditional probabilities in (\ref{Pcovb}), the overall coverage probability can be written as:
\begin{equation}
    \begin{array}{@{}r@{}l}
    &P_{\rm cov}(\beta|v_0,\psi)\\
    &\ =\displaystyle\int_0^{r_d+v_0}\P({\rm SINR}\geq \beta|r,v_0,\psi) \\
    &\ \;\;\;\;\;\;\;\;\;\;\;\;\;\;\;\;\;\;\;\;\;\;\;\;\;\;\;\;\;\;\;\;\;\;\;\;\times\displaystyle\P(E_H\geq E_{\min}|r,v_0,\psi)f_{R}(r)\dd r\\
    &\ =\displaystyle\int_0^{r_d+v_0}\exp(-r^{\alpha}[\ncalC(\tau)-\Omega_{AP}(r)-\frac{p_{BS}}{p_{AP}}\Omega_{BS}]^+)\\
    &\  \;\;\;\;\;\;\;\times\exp(-\displaystyle\frac{\beta r^{\alpha} \sigma^2}{p_{AP}})\ncalL_{AP}(\beta r^{\alpha})\ncalL_{BS}(\displaystyle\frac{p_{BS}}{p_{AP}}\beta r^{\alpha})f_{R}(r)\dd r.
    \end{array}
\end{equation}
\indent Similarly, when the effect of BSs can be ignored, the OCP is simplified as:
\begin{equation}
    \begin{array}{@{}r@{}l}
    &P_{\rm cov}(\beta|v_0)\\
    &\ =\displaystyle\int_0^{r_d+v_0}\P({\rm SINR}\geq \beta|r,v_0)\P(E_H\geq E_{\min}|r,v_0)f_{R}(r)\dd r\\
    &\ =\displaystyle\int_0^{r_d+v_0}\exp(-r^{\alpha}[\ncalC(\tau)-\Omega_{AP}(r)]^+)\\
    &\ \displaystyle \;\;\;\;\;\;\;\;\;\;\;\;\;\;\;\;\;\;\;\;\;\;\;\;\;\;\;\;\;\;\;\times \exp(-\frac{\beta r^{\alpha} \sigma^2}{p_{AP}})\ncalL_{AP}(\beta r^{\alpha})f_{R}(r)\dd r.
    \end{array}
\end{equation}

\ifCLASSOPTIONcaptionsoff
  \newpage
\fi

\bibliographystyle{IEEEtran}
\bibliography{ref}

\end{document}